\documentclass[authoryear,review,onefignum,onetabnum]{siamart190516}

\usepackage[semicolon]{natbib}
\allowdisplaybreaks
\usepackage{adjustbox}
\usepackage{enumerate}
\usepackage{bbm}
\usepackage{float}
\usepackage[shortlabels]{enumitem}

\usepackage{bm}
\usepackage{booktabs}
\usepackage{multirow}
\usepackage{threeparttable}
\usepackage{amssymb}
\usepackage{mathrsfs}
\usepackage{actuarialsymbol}

\newcommand\Item[1][i]{%
	\ifx\relax#1\relax  \item \else \item[#1] \fi
	\abovedisplayskip=0pt\abovedisplayshortskip=0pt~\vspace*{-\baselineskip}}

\usepackage{lipsum}
\usepackage{amsfonts}
\usepackage{graphicx}
\usepackage{epstopdf}
\usepackage{algorithmic}
\ifpdf
  \DeclareGraphicsExtensions{.eps,.pdf,.png,.jpg}
\else
  \DeclareGraphicsExtensions{.eps}
\fi

\newsiamremark{remark}{Remark}
\newsiamremark{hypothesis}{Hypothesis}
\crefname{hypothesis}{Hypothesis}{Hypotheses}
\newsiamthm{claim}{Claim}

\headers{Achieving a Given Financial Goal with Optimal DTI Purchasing Policy}{Y. Q. Li and L. H. Zhang}

\title{Achieving a Given Financial Goal with Optimal Deferred Term Insurance Purchasing Policy}

\author{Yuqi Li\thanks{School of Sciences, Beijing University of Posts and Telecommunications, Beijing, 100876, China 
  (\email{Llyq@bupt.edu.cn}).}
\and 
Lihua Zhang\thanks{Corresponding author. School of Sciences, Beijing University of Posts and Telecommunications, Beijing, 100876, China 
  (\email{zhlh@bupt.edu.cn}).}}

\usepackage{amsopn}

\def\bt{\begin{theorem}}
\def\et{\end{theorem}}
\def\bl{\begin{lemma}}
\def\el{\end{lemma}}
\def\br{\begin{remark}}
\def\er{\end{remark}}
\def\be{\begin{equation}}
\def\ee{\end{equation}}
\def\ce{\begin{equation*}}
\def\de{\end{equation*}}

\def\<{{\langle}}
\def\>{{\rangle}}

\newtheorem*{Main Theorem}[theorem]{Main Theorem}{\normalfont\bfseries}{\itshape}
{\bfseries}{\itshape} 
{\bfseries}{\itshape} 
{\bfseries}{\itshape} 

\nolinenumbers
\begin{document}

\maketitle

\begin{abstract}
This paper researches the problem of purchasing deferred term insurance in the context of financial planning to maximize the probability of achieving a personal financial goal. Specifically, our study starts from the perspective of hedging death risk and longevity risk, and considers the purchase of deferred term life insurance and deferred term pure endowment to achieve a given financial goal for the first time in both deterministic and stochastic framework. In particular, we consider income, consumption and risky investment in the stochastic framework, extending previous results in \cite{Bayraktar2016}. The time cutoff $m$ and $n$ make the work more difficult. However, by establishing new controls,``\emph{quasi-ideal value}" and``\emph{ideal value}", we solve the corresponding ordinary differential equations or stochastic differential equations, and give the specific expressions for the maximum probability. Then we provide the optimal life insurance purchasing strategies and the optimal risk investment strategies. In general, when $ m \geqslant 0, n>0$, deferred term insurance or term life insurance is a better choice for those who want to achieve their financial or bequest goals but are not financially sound. In particular, if $m >0, n \rightarrow \infty$, our viewpoint also sheds light on reaching a bequest goal by purchasing deferred whole life insurance. It is worth noting that when $m=0$, $ n \rightarrow \infty$, our problem is equivalent to achieving the just mentioned bequest goal by purchasing whole life insurance, at which point the maximum probability and the life insurance purchasing strategies we provide are consistent with those in \cite{Bayraktar2014, Bayraktar2016}. 
\end{abstract}



\section{Introduction}
\label{sec1a}\vspace{-4mm}
 ``Financial management" is an important issue in a person's life. Therefore, how to achieve financial goals has become one of the topics that everyone pays more and more attention to, some researches refer to \cite{MJ2021}, \cite{SM2018}. In order to achieve financial goals, many scholars have previously studied various investment strategies among different groups of people. For example, \cite{H2016} studied for college students how to do their financial planning and gave the specific steps for financial management. \cite{GG2018} subdivided the age of retirees and made a corresponding financial plan by analyzing financial capabilities and investment goals. In fact, there are many factors that affect personal financial goals. For instance, \cite{PR2015} examined a model of household financial planning that takes into account factors such as family survival, investment returns, labor income, health status, and life insurance to achieve financial goals. More about investments and financial management, please refer to \cite{Arup2017}, \cite{Biradar2021}, \cite{Bender2022}, \cite{VL2018}, \cite{KD2014}, \cite{DR2016}.\\
 \indent In last several years, \cite{DFB2018}, \cite{DS2008}, \cite{WOG2019} pointed that life insurance has attracted much attention as one of the investment methods due to its insurance characteristics, as it helps to reduce the financial burden of adverse events such as premature death, terminal illness, incapacity to work, or incapacity due to injury or disability by transferring personal losses to insurance companies. Recently, an increasing number of scholars have addressed the question of how to link the purchase of life insurance with investments.  \cite{FR1975} integrated life insurance with consumption and investment based on the optimal consumption and optimal investment problems first solved by \cite{RM1969,RM1971}. Since then, most scholars have begun to follow Richard's research methods but based on the principle of maximizing consumption utility. \cite{B2013} considered the problem about how to purchasing life insurance to maximize utility of household consumption. \cite{LZ2016} studied the stochastic optimal control problem of whole life insurance purchasing strategy with consumption and investment under the CAAR utility. \cite{L2021} considered the utility maximization problem under the exponential utility functions, and then gave the optimal portfolio, consumption and whole life insurance strategies. \\
 \indent Other than the criterion of maximizing consumption utility, many scholars have also incorporated the idea of probability into the whole life insurance purchasing strategy. They combined the minimum probability of lifetime ruin or the maximum probability of reaching a bequest goal with the whole life insurance purchasing strategy. In the early years, there were predictions that in the ten years between 2020 and 2030, the cost of living for retirees would exceed their financial capital by 400 billion dollars. Under this plan, individuals (rather than employers) must bear all investment and living risks. In this context, \cite{Y2004} gave the optimal investment strategies to minimize the probability of lifetime ruin. \cite{Bayraktar2006} used the stochastic optimal control technique to determine the optimal investment strategies for the minimum probability of lifetime ruin under the given two consumption rates and credit constraints. To make the model more complete and realistic, \cite{WY2012} introduced the annuity, and determined the minimum probability of lifetime ruin when buying a convertible life annuity and investing in a risky market. Building on the aforementioned body of prior work, Bayraktar et al examined a number of issues related to purchasing whole life insurance policy in order to maximize the probability of reaching a bequest goal. \cite{Bayraktar2014,Bayraktar2015} first considered this issue without consumption and investment, and further assumed that the force of death changes over time. Based on the above researches, \cite{Bayraktar2016} considered the model with investment and consumption, and indicated the optimal strategies for whole life insurance purchase. Recently, \cite{LY2019} determined the optimal robust strategy for maximizing the probability of reaching a bequest goal under moral hazard rate uncertainty and risky asset drift. They extended the work of Bayraktar et al by allowing that financial markets and personal mortality are ambiguous. \\
 \indent Through the study of this series of questions, previous researches have focused on the whole life insurance. However, \cite{MJ2021}, \cite{SM2018} found that due to the higher premiums of whole life insurance, it is not suitable for young or middle-aged people who need protection for a certain period of time but have only average economic power. Therefore, it is necessary to consider term life insurance. As can be seen in \cite{XS2020}, \cite{DP2011}, deferred term life insurance is a more general type of insurance than term life insurance, and from an economic perspective, it further reduces the financial burden. Thus, it is natural to include deferred term life insurance in personal financial planning. Therefore, we explore the problem of purchasing deferred term insurance in personal financial planning to maximize the probability of achieving a financial goal, and provide the optimal purchasing strategies in the deterministic framework. Furthermore, we consider the problem under the stochastic framework which include the consumption, income and risky investment, and then, give the optimal life insurance purchasing strategies and the optimal investment strategies.\\
 \indent The specific structure and innovation of this paper are as follows. In Section \ref{S2}, the policyholder purchases $m$-year deferred $n$-year term life insurance through a single premium or a continuously paid premium to take out a death benefit plan. Although the time cutoff $m$ and $n$ complicates the work, we still make progress by establishing new deterministic controls, ``\emph{quasi-ideal value}" and ``\emph{ideal value}". After solving the associated variational inequalities, we finally obtain the optimal purchasing strategies. When $ m \geqslant 0, n>0$, deferred term insurance is a better choice for those who want to achieve their financial or bequest goals but are not financially sound. In particular, if $m >0, n \rightarrow \infty$, our viewpoint also sheds light on reaching a bequest goal by purchasing deferred whole life insurance. It is worth noting that if $m=0$, $ n \rightarrow \infty$, our problem is equivalent to achieving the just mentioned bequest goal by purchasing whole life insurance. In this case, the maximum probability and life insurance purchasing strategies we provide are consistent with those in \cite{Bayraktar2014}. Moreover, whether or not the policyholder has enough time to reach the ``\emph{ideal value}" has an obvious impact on the maximum probability of achieving a given financial goal. For the case of continuous premium payment, we also study this situation by comparing the force of interest $r$ and the force of death $ \lambda$. Therefore, we give the maximum probability of achieving a given financial goal and the corresponding strategies for purchasing life insurance in the following two cases, when the policyholder has enough time to reach the ``\emph{ideal value}" (i.e. $\lambda \leqslant r$) or when there is not enough time to reach the ``\emph{ideal value}"(i.e. $\lambda >r$). In addition to the risk of death, there is another risk that deserves attention, namely the risk of longevity. To address this risk, we provide the relevant conclusions in Section \ref{S3}. \\
 \indent For the stochastic framework, we assume that the policyholder's account includes consumption, income and risky investment to achieve a financial goal under two different models in Section \ref{S4}. We introduce the approximation functions and functional operators, then by the It\^o's formula and Legendre transform, we solve the maximal probability and give the optimal $n$-year term life insurance purchasing strategies and the optimal investment strategies. When $n \rightarrow \infty$ and when the income process vanishes, our results are also consistent with those in \cite{Bayraktar2016}. Finally, Section \ref{S5} concludes the paper. From the above analyses and our results, it is clear that our research findings are more general in comparison with whole life insurance. At the same time, it lays a new direction and research foundation for personal financial planning issues. 

       \section{Purchasing  $m$-year deferred $n$-year term life insurance in personal financial planning under the deterministic framework}\label{S2}
       
       \indent Assume the policyholder has an investment account to achieve a financial goal, and this account doesn't include the consuming and risky investment. He/She may invest in a risk-free market to earn interest, which has the force of interest $r$, or purchase $m$-year deferred $n$-year term life insurance. This insurance can help the policyholder to achieve financial motivation. Let $ \delta _{d}$ be the future lifetime, it follows an exponential distribution with parameter $ \lambda $, that is, the probability density function of $\delta _{d} $ is $f(x)=\lambda e^{-\lambda x},x>0 $. Equivalently, the policyholder is subject to a constant force of mortality $ \lambda $. What needs to be explained here is that, in this paper, we always assume that the moment when the policyholder purchases insurance is the initial moment.
       \subsection{$m$-year deferred $n$-year term life insurance purchased by a single premium}
       The policyholder buys $m$-year deferred $n$-year term life insurance by a single premium with no cash value available. We define that a dollar death benefit payable immediately at time $ \delta _{d}$ which between $m$ and $ m+n$ costs $ _{m\shortmid}\textrm{K}_{n}$. The premium is payable at the moment of the contract, so $ _{m\shortmid}\textrm{K}_{n} $ is the single premium per dollar of death benefit. We give the single premium as follows:
       \begin{equation}
       _{m\shortmid}\textrm{K}_{n}=\Bigl(1+\theta \Bigr)\int_{m}^{m+n}e^{-rt}\lambda e^{-\lambda t}dt=\Bigl(1+\theta \Bigr)\frac{\lambda }{\lambda +r}\Bigl(e^{-(r+\lambda )m}-e^{-(r+\lambda )(m+n)}\Bigr), \notag
       \end{equation}
       in which $\theta$ is the proportional risk loading. \\
       \indent Denote the wealth in this investment account at time $t\geqslant 0$ by $W(t)$, and denote the amount of death benefit payable at time $ \delta _{d} $ purchased at or before time $t$ by $D(t)$, in which $ 0\leqslant t < \delta _{d}$. Therefore, with single-premium $m$-year deferred $n$-year term life insurance, the wealth follows the dynamics
       \begin{equation}
       \begin{cases}
       dW(t)=rW(t-)dt-\:_{m\shortmid}\textrm{K}_{n}dD(t),&\ 0\leqslant t<\delta _{d},\\[2mm]W(\delta _{d})=W(\delta _{d}-)+D(\delta _{d}-)\mathbbm{1}_{\left \{ m<\delta _{d}\leqslant m+n \right \}},
       \end{cases} \notag
       \end{equation}
       in which $\mathbbm{1}_{\left \{m< \delta _{d}\leqslant m+n \right \}}$ equals $1$, otherwise, it equals $0$. The similar definition in the latter sections that we omit the explanation.
       \begin{remark}
       	We define insurance purchasing strategy $\textnormal{\textbf{D}}=\left \{ D(t) \right \}_{t\geqslant 0}$. In this section, $ \textnormal{\textbf{D}}$ is called admissible if $\textnormal{\textbf{D}}$ is a non-decreasing, non-negative process, independent of $ \delta _{d} $; and if wealth under this process is non-negative for all $t\geqslant 0 $.
       \end{remark}
       
       \indent We define the maximum probability of achieving the given financial goal (id est $f$) as follows:
       \begin{equation}
       \begin{split}
       _{m \shortmid}\varphi_{n}(w,D)&=\sup_{\textbf{D}}  \biggl[\textbf{P}^{w,D}(W(\delta _{d}) \geqslant f \mid \delta_{d} < m)P(\delta_{d} < m)\\& \quad+ \textbf{P}^{w,D}(W(\delta _{d}) \geqslant f \mid m \leqslant \delta_{d} \leqslant m+n)P(m \leqslant \delta_{d} \leqslant m+n)\\& \quad +\textbf{P}^{w,D}(W(\delta _{d}) \geqslant f \mid \delta_{d} > m+n)P(\delta_{d} > m+n)  \biggr]\\
       &:=\sup_{\textbf{D}}  \bigl(\textbf{P}^{w,D}_{1}+ \textbf{P}^{w,D}_{2}+\textbf{P}^{w,D}_{3} \bigr)
       \end{split}		\notag
       \end{equation}
       in which $\textbf{P}^{w,D}$ denotes the conditional probability given $W(0-)=w$ and $D(0-)=D$. $ \textbf{P}^{w,D}_{1}=\textbf{P}^{w,D}(W(\delta _{d})\geqslant f \mid \delta_{d} <m)P(\delta_{d} <m)$, and the definitions of $ \textbf{P}^{w,D}_{i},\:i=2,3$ are similar. $D$ represents the compensation available for other types of financial products that the policyholder owns before purchasing $m$-year deferred $n$-year term life insurance, and $ D<f $, $ 0\leqslant w<( _{m\shortmid}\textrm{K}_{n}+1)(f-D):=w_{*}$. For the convenience of later discussions, we introduce the following notations
       \begin{equation}
       \begin{split}
       _{m\shortmid}\varphi_{n}(w,D):=\: _{m\shortmid}\varphi^{1}_{n}(w,D)+ \: _{m\shortmid}\varphi^{2}_{n}(w,D)+ \: _{m\shortmid}\varphi^{3}_{n}(w,D),
       \end{split}		\notag
       \end{equation}
       \begin{equation}
       \begin{split}
       _{m\shortmid}\varphi^{2,3}_{n}(w,D):=\: _{m\shortmid}\varphi^{2}_{n}(w,D)+ \: _{m\shortmid}\varphi^{3}_{n}(w,D),
       \end{split}		\notag
       \end{equation}
       in which
       \begin{equation}
       _{m\shortmid}\varphi^{i}_{n}(w,D)=\sup_{\textbf{D}} \textbf{P}^{w,D}_{i}, \; i=1,2,3.\notag
       \end{equation} 
       
       \begin{remark}
       	{\rm(1)}. In this paper, we assume that $D<f$ since the policyholder would achieve the financial goal if $ D\geqslant f $.\\
       	{\rm(2)}. If wealth equals $ _{m\shortmid}\textrm{K}_{n}(f-D):= K_{*}$, then the policyholder will spend all wealth to buy $m$-year deferred $n$-year term life insurance of $ f-D$, and if he/she survives more than $m$ years after purchasing insurance, but dies within $m+n$ years, then the policyholder's total death benefit becomes $(f-D)+D=f$, we simply call $ K_{*}$ the ``\emph{quasi-ideal value}". If the policyholder dies after $m+n$ years of purchasing insurance, he/she may not receive the death benefit. So, we need to find the real ``\emph{ideal value}". If wealth equals $w_{*}$ and policyholder doesn't receive the death benefit, then only when the remained wealth is greater than or equal to $ f-D$, the policyholder may achieve the financial goal. Therefore, $w_{*}$ is called the ``\emph{ideal value}", at this time, it's optimal for the policyholder to purchase $m$-year deferred $n$-year term life insurance of $f-D$, whether or not he/she can receive the death benefit, the total wealth at $ \delta _{d}$ will reach the given financial target $ f$. \\
       	\indent Under the purchasing strategies discussed above, we obtain $_{m\shortmid}\varphi_{n}(w,D)=1 $ for $ w\geqslant w_{*},\; 0\leqslant D< f$. Thus, all that remains is to calculate the maximum probability of achieving the given financial goal $f$ for $ \L =\left \{ (w,D):0\leqslant w<w_{*}, \; 0\leqslant D<f \right \}$. 
       \end{remark}
       
       \begin{proposition}\label{P1}
       	The maximum probability of achieving the financial goal on $ \L $ is given by 
       	\begin{equation}
       	_{m\shortmid}\varphi_{n}(w,D)=\begin{cases} 
       	\displaystyle \left ( \frac{w}{ K_{*}} \right )^{\frac{ \lambda }{r}}\biggl(e^{-\lambda m}-e^{-\lambda(m+n)}\biggr), &0\leqslant w<  K_{*},\\[3mm]
       	\displaystyle \biggl(e^{-\lambda m}-e^{-\lambda (m+n)}\biggr)+e^{-\lambda (m+n)}\left ( \frac{w- K_{*}}{f-D} \right )^{\frac{\lambda }{r}},& K_{*}\leqslant w < w_{0},\\[3mm]
       	\displaystyle \biggl(1+e^{-\lambda (m+n)}\biggr)\left ( \frac{w- K_{*}}{f-D} \right )^{\frac{\lambda }{r}}-e^{-\lambda (m+n)},&w_{0}\leqslant w < w_{*},
       	\end{cases}  \notag
       	\end{equation}
       	in which $w_{0}=(e^{-rm}+\: _{m\shortmid}\textrm{K}_{n})(f-D)$.\\
       	\indent The related optimal insurance purchasing strategy is not to purchase additional insurance until wealth reaches $w_{*}$, at which point, it is optimal to buy additional $m$-year deferred $n$-year term life insurance of $ f-D$.
       \end{proposition}
       
       \indent To prove the Proposition \ref{P1}, we first give several auxiliary lemmas.
       \begin{lemma}\label{L1}
       	Let $\phi =\phi (w,D) $ be a function that is non-decreasing, continuous, and piecewise differentiable with respect to both $w$ and $D$ on $  \L $, except that $\phi $ might have infinite derivative with respect to $w $ at $ w=0$. Suppose $ \phi $ satisfies the following variational inequality on $  \L $, except possibly when $w=0 $.\\
       	When $ 0 < w<  K_{*} $,
       	\begin{equation}
       	{\rm max} \Bigl(rw\phi _{w}-\lambda \phi ,\; \phi _{D}-\: _{m\shortmid}\textrm{K}_{n}\phi _{w}\Bigr)=0 .\label{E1}
       	\end{equation}
       	When $ w \geqslant  K_{*} $, 
       	\begin{equation}
       	{\rm max} \Bigl(rw\phi _{w}-\lambda \phi +(\lambda (e^{-\lambda m}-e^{-\lambda (m+n)})-r K_{*} \phi _{w}),\; \phi _{D}-( _{m\shortmid}\textrm{K}_{n}+1)\phi _{w} \Bigr)=0. \label{E2}
       	\end{equation}
       	Additionally, suppose 
       	\begin{equation}
       	\phi \bigl( K_{*} ,\;D\bigr)=e^{-\lambda m}-e^{-\lambda (m+n)},\quad \phi \bigl( w_{*} ,\;D\bigr)=e^{-\lambda m}. \notag
       	\end{equation}
       	Then, on $ \L$, $_{m\shortmid}\varphi^{2,3}_{n}(w,D) = \phi (w,D) $.
       \end{lemma}
       \begin{proof}
       	\rm First, we consider the case when $ 0 < w< K_{*} $. In this case, 
       	\begin{equation}
       	_{m\shortmid}\varphi^{3}_{n}(w,D)=0,  \quad _{m\shortmid}\varphi^{2,3}_{n}(w,D)=\:_{m\shortmid}\varphi^{2}_{n}(w,D). \notag
       	\end{equation}
       	\indent	Then, we rewrite the expression for $ _{m\shortmid}\varphi^{2}_{n}(w,D)$ as follows: 
       	\begin{equation}
       	_{m\shortmid}\varphi^{2}_{n}(w,D)=\sup_{\textbf{D}}\textbf{E}^{w,D}\left [ \int_{0}^{\infty}\lambda e^{-\lambda t}\mathbbm{1}_{\left \{W(t)+D(t)\geqslant f\right \}}\bigl(e^{-\lambda m}-e^{-\lambda (m+n)}\bigr)dt \right ], \notag
       	\end{equation}
       	in which $\textbf{E}^{w,D}$ denotes the conditional expectation given $W(0-)=w$, $D(0-)=D$.\\
       	\indent	Define the functional operator $ \mathscr{L}^{D}$:
       	\begin{equation}
       	\mathscr{L}^{D}\phi=rw\phi _{w}-\lambda \phi+\lambda(e^{-\lambda m}-e^{-\lambda (m+n)}) \mathbbm{1}_{\left \{ w+D\geqslant f \right \}}. \notag
       	\end{equation}
       	
       	\indent	Since we are considering a single premium, we obtain $ \mathscr{L}^{D}\phi \leqslant 0$ from the lemma hypothesis.
       	Define $ \tau _{n}=inf\left \{ s\geqslant 0:D_{s}\geqslant n \right \}$, then applying the It$\hat{o}$'s formula for $ e^{-\lambda \tau _{n}}\phi (W(\tau _{n}),D(\tau _{n}))$, we have\\
       	\begin{equation}
       	\begin{split}\label{E3}
       	e^{-\lambda \tau _{n}}\phi (W(\tau _{n}),D(\tau _{n}))&=\phi (w,D)\\&+\int_{0}^{\tau _{n}}e^{-\lambda t}(\mathscr{L}^{D}\phi-\lambda(e^{-\lambda m}-e^{-\lambda (m+n)}) \mathbbm{1}_{\left \{W(t)+D(t)\geqslant f\right \}})dt\\&+
       	\int_{0}^{\tau _{n}}e^{-\lambda t}(\phi _{D}-\: _{m\shortmid}\textrm{K}_{n}\phi _{w})dD(t).
       	\end{split} 
       	\end{equation}
       	By \cite{WY2012}, we can first assume $ \phi$ is bounded from below and after removing this assumption, the conclusion still holds. Suppose $ \phi \geqslant \phi^{*}$. 
       	By taking expectations of (\ref{E3}), we obtain
       	{\small
       	\begin{equation}
       	\begin{split}
       	\textbf{E}^{w,D}\left [ e^{-\lambda \tau _{n}} \phi ^{*}\right ] &\leqslant \textbf{E}^{w,D}\left [ e^{-\lambda \tau _{n}} \phi \right ]\\&=\phi (w,D)\\&+\textbf{E}^{w,D}\left [ \int_{0}^{\tau _{n}}e^{-\lambda t}(\mathscr{L}^{D}\phi-\lambda(e^{-\lambda m}-e^{-\lambda (m+n)}) \mathbbm{1}_{\left \{W(t)+D(t)\geqslant f\right \}})dt\right ]\\&+
       	\textbf{E}^{w,D}\left [\int_{0}^{\tau _{n}}e^{-\lambda t}(\phi _{D}- \:_{m\shortmid}\textrm{K}_{n}\phi _{w})dD(t)\right ]\\&
       	\leqslant \phi (w,D)-\textbf{E}^{w,D}\left [ \int_{0}^{\tau _{n}}\lambda(e^{-\lambda m}-e^{-\lambda (m+n)})e^{-\lambda t}\mathbbm{1}_{\left \{W(t)+D(t)\geqslant f\right \}}dt\right ].\label{E4} 
       	\end{split} 
       	\end{equation}
       }
       	Let $ \tau_{n} \rightarrow \infty$ and apply the monotonic convergence theorem to (\ref{E4})
       	\begin{equation}
       	\phi (w,D)\geqslant \:_{m\shortmid}\varphi^{2}_{n}(w,D).\notag
       	\end{equation}
       	By \cite{WY2012}, because $ \phi(w,D)$ satisfies the boundary conditions, then $\phi(w,D) =\: _{m\shortmid}\varphi^{2}_{n}(w,D) $.\\
       	\indent	For $ w \geqslant  K_{*}$, let $ \psi(w,D) =\:_{m\shortmid}\varphi^{2,3}_{n}(w,D) -(e^{-\lambda m}-e^{-\lambda (m+n)}) $ and $ w_{0}=w- K_{*} $. Then repeat the above steps, we also can obtain $\phi(w,D) = \:_{m\shortmid}\varphi^{2,3}_{n}(w,D) $.
       \end{proof}
       \begin{lemma}\label{L2}
       	The maximum probability of achieving the financial goal on $ \L$ is given by 
       	\begin{equation}
       	_{m\shortmid}\varphi^{2,3}_{n}(w,D) =\begin{cases} 
       	\displaystyle \left ( \frac{w}{ K_{*}} \right )^{\frac{ \lambda }{r}}\biggl(e^{-\lambda m}-e^{-\lambda(m+n)}\biggr), &0\leqslant w< K_{*},\\[3mm]
       	\displaystyle \biggl(e^{-\lambda m}-e^{-\lambda (m+n)}\biggr)+e^{-\lambda (m+n)}\left ( \frac{w- K_{*}}{f-D} \right )^{\frac{\lambda }{r}},& K_{*}\leqslant w < w_{*}.
       	\end{cases}  \label{E5}
       	\end{equation}
       	\indent The related optimal insurance purchasing strategy is not to purchase additional insurance until wealth reaches $w_{*}$, at which point, it is optimal to buy additional $m$-year deferred $n$-year term life insurance of $ f-D$.
       \end{lemma}
       \begin{proof}
       	\rm In order to help us solve the variational inequality in verification Lemma \ref{L1}, we recall the similar problems in \cite{MY2006}. Then in our situation, $_{m\shortmid}\varphi^{2,3}_{n}$ solves the following boundary-value problem on $\L$, except possibly at $w=0$.
       	\begin{equation}
       	\begin{cases}
       	rw(_{m\shortmid}\varphi^{2,3}_{n})_{w}-\lambda \:_{m\shortmid}\varphi^{2,3}_{n} =\mathbbm{1}_{\left \{ w \geqslant  K_{*} \right \}}\bigl(r K_{*}(_{m\shortmid}\varphi^{2,3}_{n})_{w}-\lambda (e^{-\lambda m}-e^{-\lambda (m+n)})\bigr),\\[2mm] _{m\shortmid}\varphi^{2,3}_{n}\bigl( K_{*},\;D\bigr)=e^{-\lambda m}-e^{-\lambda (m+n)},\\[2mm] _{m\shortmid}\varphi^{2,3}_{n} \bigl(w_{*}, \;D\bigr)=e^{-\lambda m}.
       	\end{cases}  \label{E6}
       	\end{equation}
       	
       	\indent	Our general guideline is to use verification Lemma \ref{L1} to prove Lemma \ref{L2}.\\
       	\indent	\rm We notice that $ _{m\shortmid}\varphi^{2,3}_{n} $ in (\ref{E5}) is increasing and differentiable with respect to both $w$ and $D$ on $ \L $, except possibly at $w=0$. $_{m\shortmid}\varphi^{2,3}_{n} $ solves the boundary-value problem (\ref{E6}), which means 
       	\begin{equation}\label{E7}
       	rw(_{m\shortmid}\varphi^{2,3}_{n})_{w}-\lambda \:_{m\shortmid}\varphi^{2,3}_{n} +\mathbbm{1}_{\left \{ w \geqslant K_{*} \right \}}\bigl(\lambda (e^{-\lambda m}-e^{-\lambda(m+n)})-r K_{*}(_{m\shortmid}\varphi^{2,3}_{n})_{w}\bigr)=0
       	\end{equation}
       	on $ \L$, except possibly at $w=0$. We give the method of solving $_{m\shortmid}\varphi^{2,3}_{n}$ as follows and it will be omitted that the similar method of solving the maximum probability in latter sections.\\
       	\indent	\rm Firstly, we solve the equation 
       	\begin{equation}
       	\lambda\:_{m\shortmid}\varphi^{2,3}_{n}= rw(_{m\shortmid}\varphi^{2,3}_{n})_{w} .\notag
       	\end{equation}
       	By rewriting $(_{m\shortmid}\varphi^{2,3}_{n})_{w} $ as  $ \frac{\partial\: _{m\shortmid}\varphi^{2,3}_{n} }{\partial w}$, equivalently
       	\begin{equation}
       	\frac{\partial \:_{m\shortmid}\varphi^{2,3}_{n} }{_{m\shortmid}\varphi^{2,3}_{n}}=\frac{\lambda }{r}\frac{\partial w}{w} . \notag
       	\end{equation}
       	Then we integrate both sides and obtain 
       	\begin{equation}
       	{\rm ln}\:_{m\shortmid}\varphi^{2,3}_{n}=\frac{\lambda }{r}{\rm ln}w +C(D) ,\notag
       	\end{equation}
       	in which $ C(D)$ is a function about $ D$. Thus, $ _{m\shortmid}\varphi^{2,3}_{n}=\tilde{C}(D)w^{\frac{\lambda }{r}}$, in which $ \tilde{C}(D)$ is also a function about $ D$. Substituting $_{m\shortmid}\varphi^{2,3}_{n} ( K_{*},D)=e^{-\lambda m}-e^{-\lambda (m+n)}$ into $ _{m\shortmid}\varphi^{2,3}_{n}=\tilde{C}(D)w^{\frac{\lambda }{r}}$, then 
       	\begin{equation}
       	\tilde{C}(D)=\bigl( K_{*}\bigr)^{-\frac{\lambda }{r}}(e^{-\lambda m}-e^{-\lambda (m+n)}) .\notag
       	\end{equation}
       	Therefore
       	\begin{equation}
       	_{m\shortmid}\varphi^{2,3}_{n}(w,D)=
       	\left ( \frac{w}{ K_{*}} \right )^{\frac{\lambda }{r}}(e^{-\lambda m}-e^{-\lambda (m+n)}) .\notag 
       	\end{equation}
       	\indent	Secondly, we solve the equation 
       	\begin{equation}
       	rw(_{m\shortmid}\varphi^{2,3}_{n})_{w}-\lambda \:_{m\shortmid}\varphi^{2,3}_{n}=r K_{*}(_{m\shortmid}\varphi^{2,3}_{n})_{w}-\lambda (e^{-\lambda m}-e^{-\lambda (m+n)})  .\notag
       	\end{equation}
       	Equivalently, we have
       	\begin{equation}
       	r(w- K_{*})(_{m\shortmid}\varphi^{2,3}_{n})_{w}=\lambda(_{m\shortmid}\varphi^{2,3}_{n}-(e^{-\lambda m}-e^{-\lambda (m+n)})).\notag
       	\end{equation}
       	Then we make variable substitution. Let $ y=w- K_{*}$ and $ _{m\shortmid}\phi^{2,3}_{n}=\:_{m\shortmid}\varphi^{2,3}_{n}-(e^{-\lambda m}-e^{-\lambda (m+n)})$. Then the equation changes to the type which we have proved, and repeat above proof, we can obtain 
       	\begin{equation}
       	_{m\shortmid}\varphi^{2,3}_{n}\bigl(w,D\bigr)=
       	\Bigl (e^{-\lambda m}-e^{-\lambda (m+n)}\Bigr)+e^{-\lambda (m+n)}\left ( \frac{w- K_{*}}{f-D} \right )^{\frac{\lambda }{r}}. \notag
       	\end{equation}     
       	\indent	Next, we prove that 
       	\begin{equation}
       	(_{m\shortmid}\varphi^{2,3}_{n})_{D}-\: _{m\shortmid}\textrm{K}_{n}(_{m\shortmid}\varphi^{2,3}_{n})_{w}< 0,\quad 0<w<  K_{*},\notag
       	\end{equation}
       	and
       	\begin{equation}
       	(_{m\shortmid}\varphi^{2,3}_{n})_{D}-( _{m\shortmid}\textrm{K}_{n}+1)(_{m\shortmid}\varphi^{2,3}_{n})_{w}< 0,\quad  K_{*}\leqslant w< w_{*}.\notag
       	\end{equation}
       	If $ 0< w<  K_{*} $, then
       	\begin{equation}
       	(_{m\shortmid}\varphi^{2,3}_{n})_{D}=\frac{\lambda }{r}\left ( \frac{w}{ K_{*}} \right )^{\frac{\lambda }{r}-1}(e^{-\lambda m}-e^{-\lambda(m+n)})\frac{w}{ _{m\shortmid}\textrm{K}_{n}(f-D)^{2}},\notag
       	\end{equation}
       	and
       	\begin{equation}
       	(_{m\shortmid}\varphi^{2,3}_{n})_{w}=\frac{\lambda }{r}\left ( \frac{w}{ K_{*}} \right )^{\frac{\lambda }{r}-1}(e^{-\lambda m}-e^{-\lambda(m+n)})\frac{1}{ K_{*}}.\notag
       	\end{equation}
       	Thus 
       	\begin{equation}
       	(_{m\shortmid}\varphi^{2,3}_{n})_{D}-\: _{m\shortmid}\textrm{K}_{n}(_{m\shortmid}\varphi^{2,3}_{n})_{w}=\frac{\lambda }{r \:_{m\shortmid}\textrm{K}_{n}}\left ( \frac{w}{ K_{*}} \right )^{\frac{\lambda }{r}-1}(e^{-\lambda m}-e^{-\lambda(m+n)})\frac{w- K_{*}}{(f-D)^{2}} < 0. \notag
       	\end{equation}
       	If $ K_{*}\leqslant w< w_{*} $, then analogy to the above method, we have 
       	\begin{equation}
       	(_{m\shortmid}\varphi^{2,3}_{n})_{D}-( _{m\shortmid}\textrm{K}_{n}+1)(_{m\shortmid}\varphi^{2,3}_{n})_{w}=\frac{\lambda }{r}\left ( \frac{w- K_{*}}{f-D} \right )^{\frac{\lambda }{r}-1}e^{-\lambda(m+n)}\frac{w-w_{*}}{(f-D)^{2}} < 0. \notag
       	\end{equation}
       	\indent	Therefore, it's shown that the expression for $ _{m\shortmid}\varphi^{2,3}_{n}$ in (\ref{E5}) satisfies the variational inequality(\ref{E1}) and (\ref{E2}). When $w=0$, obviously, $ _{m\shortmid}\varphi^{2,3}_{n}(w,D)=0$. The optimal insurance purchasing strategy is to buy additional $m$-year deferred $n$-year term life insurance of $f-D$ when wealth reaches $ w_{*}$.
       \end{proof}
       \textit{\textbf{Proof of Proposition \ref{P1}}}:   \textbf{Step1}. We first calculate $ _{m\shortmid}\varphi^{1}_{n}(w,D)$. \\
       \rm \noindent (\textbf{\romannumeral1}). If $ w \geqslant  K_{*}$, then we set $ t _{*}$ satisfies $ \bigl(w- K_{*}\bigr)e^{rt_{*}}=f-D$, the maximum probability of achieving the financial goal $ f$ is as follows: 
       \begin{equation}
       _{m\shortmid}\varphi^{1}_{n}(w,D)=\int_{t_{*}}^{m}\lambda e^{-\lambda t}dt=\biggl(\frac{w- K_{*}}{f-D}\biggr)^{\frac{\lambda }{r}}-e^{-\lambda m},\notag
       \end{equation}
       in which $ w$ satisfies $ w \geqslant (e^{-rm}+ \:_{m\shortmid}\textrm{K}_{n})(f-D) $ because of $ _{m\shortmid}\varphi^{1}_{n}(w,D) \geqslant 0 $. If $ w <  (e^{-r m}+ \:_{m\shortmid}\textrm{K}_{n})(f-D) $, then $_{m\shortmid}\varphi^{1}_{n} (w,D)=0 $.\\
       \rm \noindent (\textbf{\romannumeral2}). If $ w <  K_{*}$, then we set $ t _{*}$ satisfies $ we^{rt_{*}}=f-D$, the maximum probability of achieving the financial goal $ f$ is as follows: 
       \begin{equation}
       _{m\shortmid}\varphi^{1}_{n} (w,D)=\int_{t_{*}}^{m}\lambda e^{-\lambda t}dt=\biggl(\frac{w}{f-D}\biggr)^{\frac{\lambda }{r}}-e^{-\lambda m},\notag
       \end{equation}
       in which $ w$ satisfies $ w \geqslant e^{-rm}(f-D) $ because of $ _{m\shortmid}\varphi^{1}_{n}(w,D) \geqslant 0 $. If $ w <  e^{-r m}(f-D) $, then $_{m\shortmid}\varphi^{1}_{n} (w,D)=0 $.\\
       \textbf{Step2}. For $_{m\shortmid}\varphi^{2,3}_{n}$, the relevant conclusions have been given in the previous Lemma 2.2. Having said all of above, we restrict the premium $   _{m\shortmid}\textrm{K}_{n}<e^{-r m} $, then by Lemma \ref{L1} and Lemma \ref{L2}, the Proposition \ref{P1} is proved.$ \hfill\blacksquare $
       \begin{remark}
       	We now give another explanation of Lemma \ref{L2}. If the policyholder dies after $m$ years of purchasing insurance and the wealth reaches the ``\emph{quasi-ideal value}", the probability of the policyholder receiving the death benefit is $e^{-\lambda m}-e^{-\lambda(m+n)} $. If the policyholder dies after $m+n$ years of purchasing insurance, then he/she is only dependent on the remaining assets to achieve the financial goal by investing in a risk-free market with interest returns. The time when the wealth reaches the ``\emph{ideal value}", denoted by $ t_{0} $, is given by $\bigl(w- K_{*}\bigr)e^{rt_{0}}=f-D $. Thus $ t_{0}=\frac{1}{r}ln\bigl(\frac{f-D}{w- K_{*}}\bigr) $. Therefore the policyholder will achieve the financial goal if he/she dies after $ t_{0} $, and this occurs with probability $ e^{-\lambda t_{0}}$. In summary, the maximum probability of achieving the financial goal is $ e^{-\lambda m}-e^{-\lambda(m+n)}+e^{-\lambda t_{0}}e^{-\lambda(m+n)}$.
       \end{remark}
       
       If the policyholder purchases $n$-year term life insurance with a single premium, the maximum probability and optimal strategy can be obtained if $ m=0$ is assumed, apparently $ P(\delta_{d} >m)=1, a.s. $. 
       \begin{corollary}\label{C1}
       	When $ m=0,n>0$, the maximum probability of achieving the financial goal $ f$ on $ \L$ is given by 
       	\begin{equation}
       	_{0\shortmid}\varphi_{n}(w,D)=\begin{cases} 
       	\displaystyle \left ( \frac{w}{K_{*}} \right )^{\frac{ \lambda }{ r}}\bigl(1-e^{- \lambda n}\bigr),&0 \leqslant w<K_{*},\\[3mm]
       	\displaystyle \bigl(1-e^{-\lambda n}\bigr)+e^{-\lambda n}\left ( \frac{w-K_{*}}{f-D} \right )^{\frac{ \lambda }{r}},&K_{*} \leqslant w< w_{*},
       	\end{cases} \notag 
       	\end{equation}
       	in which $ K_{*}=\:_{0\shortmid}\textrm{K}_{n}(f-D) $, $ w_{*}=(_{0\shortmid}\textrm{K}_{n}+1)(f-D)$.	\\	
       	\indent The related optimal insurance purchasing strategy is not to purchase additional insurance until wealth reaches $ w_{*} $, at which point, it is optimal to buy  $n$-year term life insurance of $ f-D$.
       \end{corollary}
       
       When $ m=0$ and $ n\rightarrow \infty $, our problem changes to how the policyholder reach a given bequest goal by purchasing whole life insurance, the maximum probability and life insurance purchasing strategies we give are consistent with the main results in \cite{Bayraktar2014}.
       \begin{corollary}\label{C2}
       	When $ m=0$ and $ n\rightarrow \infty $, then the maximum probability of achieving the financial goal $ f$ is given by 
       	\begin{equation}
       	_{0\shortmid}\varphi_{\infty}(w,D)=\displaystyle \left ( \frac{w}{_{0\shortmid}\textrm{K}_{\infty}(f-D)} \right )^{\frac{ \lambda }{ r}}, \quad 0\leqslant w< \:_{0\shortmid}\textrm{K}_{\infty}(f-D).\notag
       	\end{equation}
       	\indent The related optimal insurance purchasing strategy is not to purchase additional insurance until wealth reaches $_{0\shortmid}\textrm{K}_{\infty}(f-D) $, at which point, it's optimal to buy whole life insurance of $ f-D$.
       \end{corollary}
       
       In particular, when $m > 0,n\rightarrow \infty $, the policyholder purchases $m$-year deferred whole life insurance to achieve the financial goal, our viewpoint also sheds light on reaching a given bequest goal by purchasing deferred whole life insurance.
       \begin{corollary}\label{C3}
       	When $m >0, n\rightarrow \infty $, then the maximum probability of achieving the financial goal is given by 
       	\begin{equation}
       	_{m\shortmid}\varphi_{\infty}(w,D)=\begin{cases} 
       	\displaystyle \left ( \frac{w}{K_{*}} \right )^{\frac{ \lambda }{r}}e^{-\lambda m}, &0\leqslant w< K_{*},\\[6mm]
       	\displaystyle e^{-\lambda m},&K_{*}\leqslant w < w_{1},\\[2mm]
       	\displaystyle  \left ( \frac{w-K_{*}}{f-D} \right )^{\frac{\lambda }{r}},&w_{1}\leqslant w < w_{*},
       	\end{cases}  \notag
       	\end{equation}
       	in which $K_{*}=\:_{m\shortmid}\textrm{K}_{\infty}(f-D)$, $w_{1}=(e^{-rm}+\:_{m\shortmid}\textrm{K}_{\infty})(f-D)$ , $ w_{*}=(_{m\shortmid}\textrm{K}_{\infty}+1)(f-D) $ and $ _{m\shortmid}\textrm{K}_{\infty}<e^{-r m} $.\\
       	\indent The related optimal purchasing strategy is not to purchase additional insurance until wealth reaches $w_{*}$, at which point, it's optimal to buy $m$-year deferred whole life insurance of $ f-D$.
       \end{corollary}
       \subsection{$m$-year deferred $n$-year term life insurance purchased by a continuously paid premium}
       Assume that the policyholder buys $m$-year deferred $n$-year term life insurance via a premium paid continuously at the rate of $_{m\shortmid}\textrm{H}_{n}$ per dollar of insurance,
       \begin{equation}
       _{m\shortmid}\textrm{H}_{n}=\frac{(1+\theta )\Ax*[m|]{\term{x}{n}}}{ \ax*{\endow{x}{m}}},\quad
       \Ax*[m|]{\term{x}{n}}=\int_{m}^{m+n}\lambda e^{-(r+\lambda )t}dt,\quad
       \Ax*{\term{x}{m}}=\int_{0}^{m}\lambda e^{-(r+\lambda )t}dt,\notag
       \end{equation}
       \begin{equation}
      \ax*{\endow{x}{m}}=\frac{1-\Ax*{\endow{x}{m}}}{r}=\frac{1-\Ax*{\term{x}{m}}-\Ax{\pureendow{x}{m}}}{r},\quad \Ax{\pureendow{x}{m}}=\int_{m}^{\infty}\lambda e^{-rn}e^{-\lambda t}dt, \notag
       \end{equation}
       in which $ \theta$ is the proportional risk loading, the explanation of the corresponding symbols see \cite{DP2011}, \cite{XS2020}.
       The proportional loading includes costs, profit, and risk margin, and reserves are established. The policyholder may change the amount of his/her insurance coverage at any time. In our time-homogeneous scenario, the policyholder in this section purchases instantaneous $m$-year deferred $n$-year term life insurance. In this case, the wealth follows the dynamics
       \begin{equation}
       \begin{cases}
       dW(t)=(rW(t)-\:_{m\shortmid}\textrm{H}_{n}D(t)\mathbbm{1}_{\left \{t \leqslant m \right \}})dt,&\ 0\leqslant t<\delta _{d},\\[2mm]
       W(\delta _{d})=W(\delta _{d}-)+D(\delta _{d}-)\mathbbm{1}_{\left \{m< \delta _{d}\leqslant m+n \right \}}.
       \end{cases}\notag
       \end{equation}
       
       An admissible insurance strategy $ \textbf{D}=\left \{ D(t) \right \}_{t\geqslant 0}$ is any non-negative process, but for all $ t \geqslant0$, the probability of $ W(t) \geqslant 0$ doesn't equal one because of the negative drift term $ _{m\shortmid}\textrm{H}_{n}D(t)$. Thus, we define $ \delta _{0}=inf\left \{ t\geqslant 0:W(t)\geqslant 0 \right \} $ and the maximum probability of achieving the financial goal $f$  
       \begin{equation}
       \begin{split}
       _{m\shortmid}\Phi_{n}(w)&=\sup_{\textbf{D}}  \biggl[\textbf{P}^{w}(W(\delta _{d} \wedge \delta _{0}) \geqslant f \mid \delta_{d} < m)P(\delta_{d} < m)\\&\quad+ \textbf{P}^{w}(W(\delta _{d} \wedge \delta _{0}) \geqslant f \mid m \leqslant \delta_{d} \leqslant m+n)P(m \leqslant \delta_{d} \leqslant m+n)\\&\quad +\textbf{P}^{w}(W(\delta _{d} \wedge \delta _{0}) \geqslant f \mid \delta_{d} > m+n)P(\delta_{d} > m+n)  \biggr]\\
       &:=\sup_{\textbf{D}}  \bigl(\textbf{P}^{w}_{1}+ \textbf{P}^{w}_{2}+\textbf{P}^{w}_{3} \bigr),
       \end{split}		\notag
       \end{equation}
       in which $\textbf{P}^{w}$ denotes the conditional probability given $W(0-)=w$. $ \textbf{P}^{w}_{1}=\textbf{P}^{w}(W(\delta _{d}\wedge \delta _{0})\geqslant f \mid \delta_{d} <m)P(\delta_{d} <m)$, and the definitions of $ \textbf{P}^{w}_{i},\:i=2,3$ are similar.
       Analogous to the previous section, we give the following notations
       \begin{equation}
       \begin{split}
       _{m\shortmid}\Phi_{n}(w):=\: _{m\shortmid}\Phi^{1}_{n}(w)+ \: _{m\shortmid}\Phi^{2}_{n}(w)+ \: _{m\shortmid}\Phi^{3}_{n}(w),
       \end{split}		\notag
       \end{equation}
       \begin{equation}
       \begin{split}
       _{m\shortmid}\Phi^{2,3}_{n}(w):=\: _{m\shortmid}\Phi^{2}_{n}(w)+ \: _{m\shortmid}\Phi^{3}_{n}(w),
       \end{split}		\notag
       \end{equation}
       in which
       \begin{equation}
       _{m\shortmid}\Phi^{i}_{n}(w)=\sup_{\textbf{D}} \textbf{P}^{w}_{i}, \; i=1,2,3.\notag
       \end{equation} 
       
       \begin{remark}
       	{\rm(1)}. We first define the ``\emph{quasi-ideal value}" and the ``\emph{ideal value}". If the wealth is equal to $\frac{_{m\shortmid}\textrm{H}_{n}f}{r+\:_{m\shortmid}\textrm{H}_{n}}:=H^{*}$ which follows from the equation $ rw=\:_{m\shortmid}\textrm{H}_{n}(f-w) $, and $ H^{*}$ is called the ``\emph{quasi-ideal value}". That's, when wealth reaches  $ H^{*}$, the policyholder purchases $m$-year deferred $n$-year term life insurance of $ f- H^{*}$ via a premium paid continuously, and if he/she survives more than $m$ years after purchasing insurance, but dies within $m+n$ years, then his/her total death benefit is $ f$. However, if the policyholder dies within $m+n$ years after purchasing insurance, then he/she cannot receive the death benefit, so in this case, the policyholder cannot achieve the financial goal $ f$. We therefore call $H^{*} $ the ``\emph{quasi-ideal value}".\\
       	{\rm(2)}. Assume the ``\emph{ideal value}" is $ w^{*} $, if wealth equals $ w^{*} $, it's optimal for the policyholder to purchase $m$-year deferred $n$-year term life insurance of $ f- H^{*} $, then whether or not he/she can receive the death benefit, he/she will achieve the financial goal $ f$. Thus, $_{m\shortmid}\Phi_{n}(w)=1 $ for $ w \geqslant w^{*}$. Derived from our setting, we obtain $ w^{*} $ by the following equation 
       	\begin{equation}
       	rw^{*}-\:_{m\shortmid}\textrm{H}_{n}\bigl(f-H^{*}\bigr)=f-w^{*} .\notag
       	\end{equation}
       	Thus, we get 
       	\begin{equation}
       	w^{*}= \frac{(r+\:_{m\shortmid}\textrm{H}_{n}+r\:_{m\shortmid}\textrm{H}_{n})f}{(r+\:_{m\shortmid}\textrm{H}_{n})(r+1)}. \notag
       	\end{equation}
       \end{remark}
       \begin{proposition}\label{P2}
       	{\rm(1)}. If $ \lambda \leqslant r $, then the maximum probability of achieving the financial goal $ f$ before ruining is given by 
       	\begin{equation}
       	_{m\shortmid}\Phi_{n}(w)=\begin{cases}
       	\displaystyle \biggl ( \frac{w}{H^{*}} \biggr )^{\frac{\lambda }{r}}\biggl(e^{-\lambda m}-e^{-\lambda (m+n)}\biggr),&0 \leqslant w< H^{*},\\[3mm]
       	\displaystyle \biggl(e^{-\lambda m}-e^{-\lambda (m+n)}\biggr)+e^{-\lambda (m+n)}\biggl ( \frac{w-H^{*}}{w^{*}-H^{*}} \biggr )^{\frac{\lambda }{r}},&H^{*} \leqslant w< w^{0},\\[3mm]
       	\displaystyle \biggl(1+e^{-\lambda (m+n)}\biggr)\biggl ( \frac{w-H^{*}}{w^{*}-H^{*}} \biggr )^{\frac{\lambda }{r}}-e^{-\lambda (m+n)},& w^{0} \leqslant w< w^{*},
       	\end{cases}	  \notag
       	\end{equation}
       	in which $ w^{0} = e^{-rm}(w^{*}-H^{*})+H^{*}$ and the initial wealth $ w \in [0,w^{*}) $.\\
       	\indent The related optimal insurance purchasing strategy is not to purchase until wealth reaches $ w^{*} $, at which point, it's optimal to buy $m$-year deferred $n$-year term life insurance of $ f-H^{*} =\frac{rf}{r+\:_{m\shortmid}\textrm{H}_{n}}$.
       	
       	{\rm(2)}. If $ \lambda > r$, then the maximum probability of achieving the financial goal $ f$ before ruining is given by 
       	\begin{equation}
       	_{m\shortmid}\Phi_{n}(w)=\begin{cases}
       	\displaystyle \biggl[1-\biggl(1-\frac{w}{H^{*}}\biggr)^{\frac{\lambda }{r+\:_{m\shortmid}\textrm{H}_{n}}}\biggr]\bigl(e^{-\lambda m}-e^{-\lambda (m+n)}\bigr),&0 \leqslant w < w^{0},\\[3mm]
       	\displaystyle \biggl ( \frac{w}{H^{*}} \biggr )^{\frac{\lambda }{r}}\biggl(e^{-\lambda m}-e^{-\lambda (m+n)}\biggr),&w^{0} \leqslant w< H^{*},\\[3mm]
       	\displaystyle \biggl(e^{-\lambda m}-e^{-\lambda (m+n)}\biggr)+e^{-\lambda (m+n)}\biggl ( \frac{w-H^{*}}{w^{*}-H^{*}} \biggr )^{\frac{\lambda }{r}},&H^{*} \leqslant w< w_{1},\\[3mm]
       	\displaystyle \biggl(1+e^{-\lambda (m+n)}\biggr)\biggl ( \frac{w-H^{*}}{w^{*}-H^{*}} \biggr )^{\frac{\lambda }{r}}-e^{-\lambda (m+n)},& w_{1} \leqslant w< w^{*},
       	\end{cases}	 \notag
       	\end{equation}
       	in which $w_{1} = e^{-rm}(w^{*}-H^{*})+H^{*}$, and the initial wealth $ w \in [0,w^{*}) $, where $ w^{0}$ is the unique zero in $ (0,H^{*})$ of the following equations 
       	\begin{equation}
       	\displaystyle	1-\biggl(1-\frac{w}{H^{*}}\biggr)^{\frac{\lambda }{r+\:_{m\shortmid}\textrm{H}_{n}}}= \Bigl ( \frac{w}{H^{*}} \Bigr )^{\frac{\lambda }{r}}. \notag
       	\end{equation}
 The related optimal purchasing strategy is:
 \begin{itemize}
       \item If wealth $ w $ is less than $ w ^{0}$, then the policyholder purchase $m$-year deferred $n$-year term life insurance of $ f-w$;
     \item  If wealth $ w$ is greater than or equal to $ w^{0}$, then the policyholder doesn't purchase insurance until the wealth reaches $ w^{*}$, at which point, it's optimal to buy $m$-year deferred $n$-year term life insurance of $ f-H^{*} =\frac{rf}{r+\:_{m\shortmid}\textrm{H}_{n}} $.
       \end{itemize}
       \end{proposition}
       
       \indent To prove Proposition \ref{P2}, we first give some auxiliary lemmas.
       \begin{lemma}\label{L3}
       	Let $ F=F(w)$ be a function that is non-decreasing, continuous, and piecewise differentiable on $ [0,w^{*}) $, except that $ F $ might not be differentiable at 0. Suppose $F$ satisfies the following variational inequality on $ (0,w^{*}) $ 
       	\begin{equation}
       	\begin{split}
       	\lambda F =rwF _{w}&+\left \{ {\rm max}\bigl(\lambda (e^{-\lambda m}-e^{-\lambda (m+n)})-\:_{m\shortmid}\textrm{H}_{n}(f-w)F _{w} ,0 \bigr)\right \}\mathbbm{1}_{\left \{ w< H^{*} \right \}}\\&+\bigl(\lambda (e^{-\lambda m}-e^{-\lambda (m+n)})-rH^{*}F _{w}\bigr)\mathbbm{1}_{\left \{ w\geqslant H^{*} \right \}}, \label{E8}
       	\end{split}
       	\end{equation}
       	Additionally suppose 
       	\begin{equation}
       	F(H^{*}) =e^{-\lambda m}-e^{-\lambda (m+n)} ,\quad F(w^{*})=e^{-\lambda m}.\notag
       	\end{equation}
       	Then on $ (0,w^{*}) $, $ _{m\shortmid}\Phi^{2,3}_{n}(w)=F(w) $.
       \end{lemma}
       \begin{proof}
       	\rm	Firstly, we give the control equation about $ _{m\shortmid}\Phi^{2,3}_{n}$ as follows: 
       	\begin{equation}
       	\begin{split}
       	\lambda\: _{m\shortmid}\Phi^{2,3}_{n}& =rw(_{m\shortmid}\Phi^{2,3}_{n})_{w}+\max_{\textbf{D}}\left\{ \lambda (e^{-\lambda m}-e^{-\lambda (m+n)})\mathbbm{1}_{\left \{ w+D\geqslant f \right \}}
       	+\right.\\
       	&\phantom{=\;\;}\left.	(\:_{m\shortmid}\textrm{H}_{n}D-r(f-D))(_{m\shortmid}\Phi^{2,3}_{n})_{w}\mathbbm{1}_{\left \{ w\geqslant H^{*} \right \}}-\:_{m\shortmid}\textrm{H}_{n}D\:(_{m\shortmid}\Phi^{2,3}_{n})_{w} \right\}. \notag
       	\end{split}
       	\end{equation}
       	\noindent (\textbf{\romannumeral1}). If wealth is less than $ H^{*}$, then the control equation can be replaced with the following equivalent expression
       	\begin{equation}
       	\lambda \:_{m\shortmid}\Phi^{2,3}_{n} =rw(_{m\shortmid}\Phi^{2,3}_{n})_{w}+\max_{\textbf{D}}\left \{ \lambda (e^{-\lambda m}-e^{-\lambda (m+n)})\mathbbm{1}_{\left \{ w+D\geqslant f \right \}}-\:_{m\shortmid}\textrm{H}_{n}D\:(_{m\shortmid}\Phi^{2,3}_{n})_{w}\right \} .\notag
       	\end{equation}
       	Because of the term $-\:_{m\shortmid}\textrm{H}_{n}D\:\Phi _{w} $, we choose $ D$ to be a minimum, that is, constrained by $\mathbbm{1}_{\left \{ w+D\geqslant f \right \}} $. Equivalently, if the indicator equals 0, then the optimal insurance is $ D=0$; if the indicator equals 1, then the optimal insurance is $ D=f-w$. In this case, we have the following variational inequality 
       	\begin{equation}
       	\lambda \:_{m\shortmid}\Phi^{2,3}_{n} =rw(_{m\shortmid}\Phi^{2,3}_{n})_{w}+{\rm max}\left \{ \lambda (e^{-\lambda m}-e^{-\lambda (m+n)})-\:_{m\shortmid}\textrm{H}_{n}(f-w)(_{m\shortmid}\Phi^{2,3}_{n})_{w} ,0 \right \}.\notag
       	\end{equation}
       	\noindent (\textbf{\romannumeral2}). If wealth is greater than or equal to $ H^{*}$, then the control equation can be replaced with the following equivalent expression 
       	\begin{equation}
       	\begin{split}
       	\lambda \:_{m\shortmid}\Phi^{2,3}_{n}&=rw(_{m\shortmid}\Phi^{2,3}_{n})_{w}\\&+\max_{\textbf{D}}\left \{ \lambda (e^{-\lambda m}-e^{-\lambda (m+n)})\mathbbm{1}_{\left \{ w+D\geqslant f \right \}}-r(f-D)(_{m\shortmid}\Phi^{2,3}_{n})_{w}\right \}.\notag
       	\end{split}
       	\end{equation}
       	In this case, the term $ r(f-D)(_{m\shortmid}\Phi^{2,3}_{n})_{w}$ is non-negative, so we choose $ D$ to be a maximum, which means the optimal insurance is $ D=f- H^{*}  $, we have the following equation 
       	\begin{equation}
       	\lambda\:_{m\shortmid}\Phi^{2,3}_{n}=rw(_{m\shortmid}\Phi^{2,3}_{n})_{w}+\lambda (e^{-\lambda m}-e^{-\lambda (m+n)})-rH^{*}\:(_{m\shortmid}\Phi^{2,3}_{n})_{w} .\label{E9}
       	\end{equation}	
       	\indent	\rm Secondly, we consider the case when the wealth is less than $H^{*}$. In this case, 
       	\begin{equation}
       	_{m\shortmid}\Phi^{3}_{n}(w)=0,  \quad _{m\shortmid}\Phi^{2,3}_{n}(w)=\:_{m\shortmid}\Phi^{2}_{n}(w). \notag
       	\end{equation}
       	Rewrite the express for $ _{m\shortmid}\Phi^{2}_{n}$ as follows: \\
       	\begin{equation}
       	_{m\shortmid}\Phi^{2}_{n}(w)=\sup_{\textbf{D}}\textbf{E}^{w}\left [ \int_{0}^{\delta _{0}}\lambda e^{-\lambda t}(e^{-\lambda m}-e^{-\lambda (m+n)})\mathbbm{1}_{\left \{W(t)+D(t)\geqslant f\right \}}dt \right ] , \notag
       	\end{equation} 
       	in which $\textbf{E}^{w}$ denotes the conditional expectation given $W(0-)=w$.	\\
       	\indent Define the functional operator $ \mathscr{L}^{D}$:
       	\begin{equation}
       	\mathscr{L}^{D}F=(rw-\:_{m\shortmid}\textrm{H}_{n}D)F _{w}-\lambda F+\lambda(e^{-\lambda m}-e^{-\lambda (m+n)}) \mathbbm{1}_{\left \{ w+D\geqslant f \right \}}. \notag
       	\end{equation}
       	Since the lemma hypothesis, we can obtain $ \mathscr{L}^{D}F \leqslant 0$ .\\
   \indent Define $ \tau _{n}^{a}=inf\left \{ s\geqslant 0:D_{s}\geqslant n \right \}$, $ \tau _{n}=\tau _{n}^{a}\wedge \delta _{0}$, then applying the It$\hat{o}$'s formula for $ e^{-\lambda \tau _{n}}F(W(\tau _{n}))$, we have
       	\begin{equation}
       	e^{-\lambda \tau _{n}}F (W(\tau _{n}))=F(w)+\int_{0}^{\tau _{n}}e^{-\lambda t}(\mathscr{L}^{D}F-\lambda(e^{-\lambda m}-e^{-\lambda (m+n)}) \mathbbm{1}_{\left \{W(t)+D(t)\geqslant f\right \}})dt. \label{E10} 
       	\end{equation}
       	By \cite{WY2012}, we can first assume $ F$ is bounded from below and after removing this assumption, the conclusion still holds. Let's suppose $ F \geqslant F^{*}$. 
       	By taking expectations of (\ref{E10}), we obtain
       	\begin{equation}
       	\textbf{E}^{w}\left [ e^{-\lambda \tau _{n}} F ^{*}\right ]\leqslant F (w)-\textbf{E}^{w}\left [ \int_{0}^{\tau _{n}}\lambda(e^{-\lambda m}-e^{-\lambda (m+n)})e^{-\lambda t}\mathbbm{1}_{\left \{W(t)+D(t)\geqslant f\right \}}dt\right ].\label{E11} 
       	\end{equation}
       	Let $ \tau_{n} \rightarrow \infty$ and apply the monotonic convergence theorem to (\ref{E11})
       	\begin{equation}
       	F(w)\geqslant \:_{m\shortmid}\Phi^{2}_{n}(w).\notag
       	\end{equation}
       	Because $ F$ satisfies the boundary conditions, then $F(w) = \:_{m\shortmid}\Phi^{2}_{n}(w)$.
       	For $ w \geqslant  H^{*}$, then repeat the above steps, we also can obtain $F(w)=\: _{m\shortmid}\Phi^{2,3}_{n}(w)$ in this case.
       \end{proof}
       
       \indent Next, we compare r with $\lambda$. The above control equations yield some conclusions to be used later.\\
       \noindent (\textbf{\romannumeral1}). If wealth is less than $ H^{*}$, then at any wealth level $ w$, the policyholder chooses either to buy no insurance or to buy $m$-year deferred $n$-year term life insurance of $ f-w$. \\
       \indent Firstly, we assume that in a neighborhood of $H^{*}$, the policyholder buys $m$-year deferred $n$-year term life insurance of $ f-w$, then we need to solve the equation
       \begin{equation}
       \begin{cases}
       \lambda \:_{m\shortmid}\Phi^{2,3}_{n}=rw(_{m\shortmid}\Phi^{2,3}_{n})_{w}+\lambda (e^{-\lambda m}-e^{-\lambda (m+n)})-\:_{m\shortmid}\textrm{H}_{n}(f-w)(_{m\shortmid}\Phi^{2,3}_{n})_{w},\\[2mm]
       _{m\shortmid}\Phi^{2,3}_{n}\bigl(H^{*}\bigr) =e^{-\lambda m}-e^{-\lambda (m+n)}.
       \end{cases}\notag
       \end{equation}
       We denote the solution of this equation by $ _{m\shortmid}\Phi^{2,3}_{a,n}$, then $_{m\shortmid}\Phi^{2,3}_{a,n}$ is given as follows: 
       \begin{equation}
       _{m\shortmid}\Phi^{2,3}_{a,n}(w)=\biggl[1-k\biggl(1-\frac{w}{H^{*}}\biggr)^{\frac{\lambda }{\lambda +\:_{m\shortmid}\textrm{H}_{n}}}\biggr]\biggl(e^{-\lambda m}-e^{-\lambda (m+n)}\biggr) ,\notag
       \end{equation} 
       in which $k$ is some constant, $k>0$. \\
       \indent Secondly, we suppose that in a neighborhood of $H^{*} $, the policyholder buys no insurance, then we need to solve the equation 
       \begin{equation}
       \begin{cases}
       \lambda\: _{m\shortmid}\Phi^{2,3}_{n} =rw(_{m\shortmid}\Phi^{2,3}_{n})_{w},\\
       _{m\shortmid}\Phi^{2,3}_{n}\bigl(H^{*} \bigr) =e^{-\lambda m}-e^{-\lambda (m+n)}.
       \end{cases}\notag
       \end{equation}
       We denote the solution of this equation by $_{m\shortmid}\Phi^{2,3}_{b,n} $, then $ _{m\shortmid}\Phi^{2,3}_{b,n}$ is given as follows: 
       \begin{equation}
       _{m\shortmid}\Phi^{2,3}_{b,n}(w)=\biggl(e^{-\lambda m}-e^{-\lambda (m+n)}\biggr)\biggl(\frac{w}{H^{*}}\biggr)^{\frac{\lambda }{r}}. \notag
       \end{equation}
       \noindent (\textbf{\romannumeral2}). If wealth is greater than or equal to $H^{*} $, then the optimal insurance is that the policyholder buys $m$-year deferred $n$-year term life insurance of $ f-H^{*} $, we just need to solve the equation(\ref{E9}).\\
       \noindent (\textbf{\romannumeral3}). By the above analyses, we need to determine which of $ _{m\shortmid}\Phi^{2,3}_{a,n} $ and $ _{m\shortmid}\Phi^{2,3}_{b,n} $ is greater for $ w$ near $ H^{*} $, spontaneously, there are two cases to discuss: $ \lambda \leqslant r$ and $ \lambda > r$.
       \begin{itemize}
\item If $\lambda \leqslant r $, that means the policyholder has enough time to reach the ``\emph{ideal value}", then $_{m\shortmid}\Phi^{2,3}_{a,n}  \leqslant \:_{m\shortmid}\Phi^{2,3}_{b,n} $ for all $ 0 \leqslant w \leqslant w^{*}  $. If wealth is less than $ H^{*} $, the maximum probability of achieving the financial goal is equal to $ \bigl( \frac{w}{H^{*} } \bigr)^{\frac{\lambda }{r}}$ multiply $P( m<\delta _{d}\leqslant m+n ) $. The method for solving $  \bigl( \frac{w}{H^{*} } \bigr)^{\frac{\lambda }{r}}$ as in \cite{Bayraktar2014}. If wealth is greater than or equal to $ H^{*}  $, the policyholder purchases $m$-year deferred $n$-year term life insurance via a premium paid continuously, and if  the policyholder survives more than $m$ years after purchasing insurance, but dies within $m+n$ years, then he/she can get the death benefit and achieve the financial goal. But in another case where the policyholder cannot receive the death benefit, we need to find time $ t_{1}$ which satisfies$ (w-H^{*} )e^{rt_{1}}=w^{*}-H^{*} $, and the probability of reaching the "\emph{ideal value}" before dying is equal to $ e^{-\lambda t_{1}} $, so, in this case, the probability of achieving the financial goal is $ e^{-\lambda t_{1}}e^{-\lambda (m+n)} $;
\item  If $ \lambda > r$, then there's a wealth level $ w^{0}$ such that $ \displaystyle _{m\shortmid}\Phi^{2,3}_{n}  $ includes $ \displaystyle _{m\shortmid}\Phi^{2,3}_{a,n}  $ and $  _{m\shortmid}\Phi^{2,3}_{b,n}  $, and if the initial wealth $ w  $ is small enough (of course, $ w <H^{*}  $), we set $ w<w^{0}$. In this situation, the policyholder has no enough time to reach the "\emph{ideal value}". Therefore, the optimal purchasing strategy is to purchase $m$-year deferred $n$-year term life insurance of $ f-w$. The wealth at time $ t$ satisfies the following equation 
       \begin{equation}
       H^{*} -W(t)=\bigl(H^{*} -w\bigr)e^{(r+\:_{m\shortmid}\textrm{H}_{n})t} .\notag
       \end{equation}
       Let $ t_{2}$ satisfy that $W(t_{2})=0 $, then the probability that the policyholder achieves the financial goal equals $ (1-e^{-\lambda t_{2}})(e^{-\lambda m}-e^{-\lambda (m+n)})$. 
       If the initial wealth $ w \geqslant w^{0}$, then the results as same as $ \lambda \leqslant r$. 
    \end{itemize}

      \indent For more details, see Lemma \ref{L4} and Lemma \ref{L5}.
       \begin{lemma}\label{L4}
       	If $ \lambda \leqslant r $, then the maximum probability of achieving the financial goal $ f$ before ruining is given by 
       	\begin{equation}
       	_{m\shortmid}\Phi^{2,3}_{n}(w)=\begin{cases}
       	\displaystyle \biggl ( \frac{w}{H^{*}} \biggr )^{\frac{\lambda }{r}}\bigl(e^{-\lambda m}-e^{-\lambda (m+n)}\bigr),&0 \leqslant w< H^{*},\\[3mm]
       	\displaystyle \bigl(e^{-\lambda m}-e^{-\lambda (m+n)}\bigr)+e^{-\lambda (m+n)}\biggl ( \frac{w-H^{*}}{w^{*}-H^{*}} \biggr )^{\frac{\lambda }{r}},&H^{*} \leqslant w< w^{*},
       	\end{cases}	  \label{E12}
       	\end{equation}
       	for initial wealth $ w \in [0,w^{*}) $.\\
       	\indent The related optimal insurance purchasing strategy is not to purchase until wealth reaches $ w^{*} $, at which point, it's optimal to buy $m$-year deferred $n$-year term life insurance of $ f-H^{*} =\frac{rf}{r+\:_{m\shortmid}\textrm{H}_{n}}$.
       \end{lemma}
       \begin{proof}
       	\rm Our general guideline is to use verification Lemma \ref{L3} to prove this lemma.\\
       	\indent Firstly, we notice that $_{m\shortmid}\Phi^{2,3}_{n}$ in (\ref{E12}) is continuous and increasing on $ [0,w^{*})$, and is piecewise differentiable on $ (0,w^{*})$, obviously, we can verify that $ _{m\shortmid}\Phi^{2,3}_{n}$ satisfies the variational inequality (\ref{E8}) when $ H^{*} \leqslant w< w^{*} $.\\
       	\indent Then, we prove that the inequality $\lambda (e^{-\lambda m}-e^{-\lambda (m+n)})-\:_{m\shortmid}\textrm{H}_{n}(f-w)(_{m\shortmid}\Phi^{2,3}_{n})_{w} \leqslant 0 $ holds on $ (0,H^{*}) $.
       	By calculation we have  
       	\begin{equation}
       	\begin{split}
       	\lambda (e^{-\lambda m}-e^{-\lambda (m+n)})&-\:_{m\shortmid}\textrm{H}_{n}(f-w)(_{m\shortmid}\Phi^{2,3}_{n} )_{w}=\lambda (e^{-\lambda m}-e^{-\lambda (m+n)})\\&-\frac{\lambda (r+\:_{m\shortmid}\textrm{H}_{n})(f-w)\:_{m\shortmid}\textrm{H}_{n}}{r\:_{m\shortmid}\textrm{H}_{n}f}(e^{-\lambda m}-e^{-\lambda (m+n)})\biggl ( \frac{w}{H^{*}} \biggr )^{\frac{\lambda }{r}-1} ,\notag
       	\end{split}
       	\end{equation}
       	thus, $\lambda (e^{-\lambda m}-e^{-\lambda (m+n)})-\:_{m\shortmid}\textrm{H}_{n}(f-w)(_{m\shortmid}\Phi^{2,3}_{n})_{w} \leqslant 0 $ is equivalent to 
       	\begin{equation}
       	1-\frac{r+\:_{m\shortmid}\textrm{H}_{n}}{r}\left ( \frac{w}{H^{*}} \right )^{\frac{\lambda }{r}-1}+\frac{_{m\shortmid}\textrm{H}_{n}}{r}\left ( \frac{w}{H^{*}} \right )^{\frac{\lambda }{r}}\leqslant 0 . \notag
       	\end{equation}
       	Let $ y=\frac{w}{H^{*}} $, obviously $ y \in (0,1)$. The above inequality is equivalent to 
       	\begin{equation}
       	1-\frac{r+\:_{m\shortmid}\textrm{H}_{n}}{r}y^{\frac{\lambda }{r}-1}+\frac{_{m\shortmid}\textrm{H}_{n}}{r}y^{\frac{\lambda }{r}}\leqslant 0.\notag
       	\end{equation}
       	Referring to \cite{Bayraktar2014}, this inequality holds on $y \in (0,1) $ obviously.\\
       	\indent Therefore, we have proved that $_{m\shortmid}\Phi^{2,3}_{n} $ in (\ref{E12}) satisfies the variational inequality (\ref{E8}). When $w=0$, obviously, $ _{m\shortmid}\Phi^{2,3}_{n}(w,D)=0$. The optimal insurance strategy is not to purchase until wealth reaches $ w^{*} $, at which time, it's optimal to buy $m$-year deferred $n$-year term life insurance of $ f-H^{*} =\frac{rf}{r+\:_{m\shortmid}\textrm{H}_{n}}$. 
       \end{proof}
       
       \begin{lemma}\label{L5}
       	If $ \lambda > r$, then the maximum probability of achieving the financial goal $ f$ before ruining is given by 
       	\begin{equation}
       	_{m\shortmid}\Phi^{2,3}_{n}(w)=\begin{cases}
       	\displaystyle \biggl[1-\biggl(1-\frac{w}{H^{*}}\biggr)^{\frac{\lambda }{r+\:_{m\shortmid}\textrm{H}_{n}}}\biggr]\biggl(e^{-\lambda m}-e^{-\lambda (m+n)}\biggr),&0 \leqslant w < w^{0},\\[3mm]
       	\displaystyle \biggl ( \frac{w}{H^{*}} \biggr )^{\frac{\lambda }{r}}\biggl(e^{-\lambda m}-e^{-\lambda (m+n)}\biggr),&w^{0} \leqslant w< H^{*},\\[3mm]
       	\displaystyle \biggl(e^{-\lambda m}-e^{-\lambda (m+n)}\biggr)+e^{-\lambda (m+n)}\Biggl ( \frac{w-H^{*}}{w^{*}-H^{*}} \Biggr )^{\frac{\lambda }{r}},&H^{*} \leqslant w< w^{*},
       	\end{cases}	 \label{E13}
       	\end{equation}
       	for initial wealth $ w \in [0,w^{*}) $, where $ w^{0}$ is the unique zero in $ (0,H^{*})$ of the following equations 
       	\begin{equation}
       	1-\biggl(1-\frac{w}{H^{*}}\biggr)^{\frac{\lambda }{r+\:_{m\shortmid}\textrm{H}_{n}}}=\biggl ( \frac{w}{H^{*}} \biggr )^{\frac{\lambda }{r}}. \label{E14}
       	\end{equation}
The associated optimal purchasing strategy is:
       	\begin{itemize}
       	\item If wealth $ w $ is less than $ w ^{0}$, then the policyholder purchases $m$-year deferred $n$-year term life insurance of $ f-w$;
       	\item If wealth $ w$ is greater than or equal to $ w^{0}$, then the policyholder doesn't purchase insurance until the wealth reaches $ w^{*}$, at which point, it's optimal to buy $m$-year deferred $n$-year term life insurance of $ f-H^{*} =\frac{rf}{r+\:_{m\shortmid}\textrm{H}_{n}} $.
       \end{itemize}
       \end{lemma}
       \begin{proof}
       	\rm As similar to our previous proof, our general guideline is to use verification Lemma \ref{L3} to prove this proposition.
       	Firstly, we notice that $ _{m\shortmid}\Phi^{2,3}_{n}$ in (\ref{E13}) is continuous and increasing on $ [0,w^{*}]$ and is piecewise differentiable on $(0,w^{*})$, obviously, we can verify that $_{m\shortmid}\Phi^{2,3}_{n}$ satisfies the variational inequality (\ref{E8}) when $ H^{*}\leqslant w<w^{*} $.
       	Next, we refer to the ``Lemma 3.4 and Proposition 3.5" in \cite{Bayraktar2014} to prove the following three things:\\
       	(a). The equation (\ref{E14}) has a unique zero on $(0,H^{*}) $ that is $ w^{0}$;\\
       	(b). On $ [0,w^{0}) $, $ \lambda\:_{m\shortmid}\Phi^{2,3}_{n}=rw(_{m\shortmid}\Phi^{2,3}_{n})_{w}+\lambda (e^{-\lambda m}-e^{-\lambda (m+n)})-\:_{m\shortmid}\textrm{H}_{n}(f-w)(_{m\shortmid}\Phi^{2,3}_{n})_{w} $ and the inequality $ \lambda (e^{-\lambda m}-e^{-\lambda (m+n)})-\:_{m\shortmid}\textrm{H}_{n}(f-w)(_{m\shortmid}\Phi^{2,3}_{n})_{w} \geqslant 0 $ holds, then the optimal purchasing strategy is to buy $m$-year deferred $n$-year term life insurance of $ f-w $;\\
       	(c). On $  [w^{0},H^{*}) $ , $ \lambda\: _{m\shortmid}\Phi^{2,3}_{n}=rw(_{m\shortmid}\Phi^{2,3}_{n})_{w}$ and the inequality $ \lambda (e^{-\lambda m}-e^{-\lambda (m+n)})-\:_{m\shortmid}\textrm{H}_{n}(f-w)(_{m\shortmid}\Phi^{2,3}_{n})_{w}\\ \leqslant 0 $  holds, then the optimal strategy is not to buy insurance until the wealth reaches $ w^{*} $.\\
       	The details of above three things are the straightforward application of \cite{Bayraktar2014}.
       	Therefore, $ _{m\shortmid}\Phi^{2,3}_{n}$ in (\ref{E13}) satisfies the variational inequality (\ref{E8}). When $w=0$, obviously, $ _{m\shortmid}\Phi^{2,3}_{n}(w,D)=0$, and our conclusions hold.
       \end{proof}       
       \textit{\textbf{Proof of Proposition \ref{P2}}}:
       \rm \textbf{Step1}. We first calculate $ _{m\shortmid}\Phi^{1}_{n} $.\\
       \noindent (\textbf{\romannumeral1}). If $ \lambda \leqslant r$, $ w \geqslant H^{*}$, then we set $ t _{*}$ satisfies $ \bigl(w-H^{*}\bigr)e^{rt_{*}}=w^{*}-H^{*}$, the maximum probability of achieving the financial goal $ f$ is as follows: 
       \begin{equation}
       _{m\shortmid}\Phi^{1}_{n}(w)=\int_{t_{*}}^{m}\lambda e^{-\lambda t}dt=\biggl(\frac{w-H^{*}}{w^{*}-H^{*}}\biggr)^{\frac{\lambda }{r}}-e^{-\lambda m},\notag
       \end{equation}
       in which $ w$ satisfies $ w \geqslant e^{-rm}(w^{*}-H^{*})+H^{*} $ because of $_{m\shortmid}\Phi^{1}_{n}(w) \geqslant 0 $. If $ w <  e^{-rm}(w^{*}-H^{*})+H^{*} $, then $_{m\shortmid}\Phi^{1}_{n}(w)=0 $.\\
       \noindent (\textbf{\romannumeral2}). If $ \lambda \leqslant r$, $ w < H^{*}$, then we set $ t _{*}$ satisfies $ we^{rt_{*}}=w^{*}$, the maximum probability of achieving the financial goal $ f$ is as follows: 
       \begin{equation}
       _{m\shortmid}\Phi^{1}_{n}(w)=\int_{t_{*}}^{m}\lambda e^{-\lambda t}dt=\biggl(\frac{w}{w^{*}}\biggr)^{\frac{\lambda }{r}}-e^{-\lambda m},\notag
       \end{equation}
       in which $ w$ satisfies $ w \geqslant e^{-rm}w^{*} $ because of $ _{m\shortmid}\Phi^{1}_{n}(w) \geqslant 0 $. If $ w <  e^{-r m}w^{*} $, then $_{m\shortmid}\Phi^{1}_{n} (w)=0 $, when $ w \geqslant $\\
       \noindent (\textbf{\romannumeral3}). If $ \lambda >r$, $ w < H^{*}$, in this case, obviously, $ _{m\shortmid}\Phi^{1}_{n}(w)=0$. \\
       \noindent (\textbf{\romannumeral4}). If $ \lambda >r$, $ w \geqslant H^{*}$, the results are same as (\textbf{\romannumeral1}). \\
       \textbf{Step2}. For $_{m\shortmid}\Phi^{2,3}_{n}$, the relevant conclusions have been given in the previous Lemma \ref{L3}, Lemma \ref{L4} and Lemma \ref{L5}, as same as the previous section, we generally assume $e^{-rm} > \frac{_{m\shortmid}\textrm{H}_{n}+r\:_{m\shortmid}\textrm{H}_{n}}{r+\:_{m\shortmid}\textrm{H}_{n}+r\:_{m\shortmid}\textrm{H}_{n}}$, then the Proposition \ref{P2} is proved.$ \hfill\blacksquare $\\
       \indent If the policyholder purchases $n$-year term life insurance through a continuously paid premium, the maximum probability and the optimal strategies can be obtained by taking $ m=0$, apparently $ P(\delta_{d} >m)=1, a.s. $. We have the following Corollary \ref{C4}.
       \begin{corollary}\label{C4}
       	{\rm(1)}. When $ m=0,n>0$ and if $ \lambda \leqslant r $, then the maximum probability of achieving the financial goal $ f$ before ruining is given by 
       	\begin{equation}
       	_{0\shortmid}\Phi_{n}(w)=\begin{cases}
       	\displaystyle \biggl ( \frac{w}{H^{*}} \biggr )^{\frac{\lambda }{r}}\biggl(1-e^{-\lambda n}\biggr),&0 \leqslant w< H^{*},\\[3mm]
       	\displaystyle \biggl(1-e^{-\lambda n}\biggr)+e^{-\lambda n}\biggl ( \frac{w-H^{*}}{w^{*}-H^{*}} \biggr )^{\frac{\lambda }{r}},&H^{*} \leqslant w< w^{*},
       	\end{cases}	\notag
       	\end{equation}
       	for initial wealth $ w \in [0,w^{*}) $, in which $H^{*}=\frac{_{0\shortmid}\textrm{H}_{n}f}{r+\:_{0\shortmid}\textrm{H}_{n}} $, $w^{*}=\frac{(r+\:_{0\shortmid}\textrm{H}_{n}+r\:_{0\shortmid}\textrm{H}_{n})f}{(r+\:_{0\shortmid}\textrm{H}_{n})(r+1)}$.\\
       	\indent The related optimal insurance purchasing strategy is not to purchase until wealth reaches $ w^{*} $, at which point, it's optimal to buy $n$-year term life insurance of $ f-H^{*} =\frac{rf}{r+\:_{0\shortmid}\textrm{H}_{n}}$.\\
       	{\rm(2).} When $ m=0, n>0$ and if $ \lambda > r$, then the maximum probability of achieving the financial goal $ f$ before ruining is given by 
       	\begin{equation}
       	_{0\shortmid}\Phi_{n}(w)=\begin{cases}
       	\displaystyle \biggl[1-\biggl(1-\frac{w}{H^{*}}\biggr)^{\frac{\lambda }{r+\:_{0\shortmid}\textrm{H}_{n}}}\biggr]\biggl(1-e^{-\lambda n}\biggr),&0 \leqslant w < w^{0},\\[3mm]
       	\displaystyle \biggl ( \frac{w}{H^{*}} \biggr )^{\frac{\lambda }{r}}\biggl(1-e^{-\lambda n}\biggr),&w^{0} \leqslant w< H^{*},\\[3mm]
       	\displaystyle 1-e^{-\lambda n}+e^{-\lambda n}\biggl ( \frac{w-H^{*}}{w^{*}-H^{*}} \biggr )^{\frac{\lambda }{r}},&H^{*} \leqslant w< w^{*},
       	\end{cases}	 \notag
       	\end{equation}
       	for initial wealth $ w \in [0,w^{*}) $, where $ w^{0}$ is the unique zero in $ (0,H^{*})$ of the following equations 
       	\begin{equation}
       	1-\biggl(1-\frac{w}{H^{*}}\biggr)^{\frac{\lambda }{r+\:_{0\shortmid}\textrm{H}_{n}}}=\biggl ( \frac{w}{H^{*}} \biggr )^{\frac{\lambda }{r}}. \notag
       	\end{equation}
The related optimal purchasing strategy is
       	\begin{itemize}
       	\item If wealth $ w $ is less than $ w ^{0}$, then the policyholder purchases $n$-year term life insurance of $ f-w$;
       \item If wealth $ w$ is greater than or equal to $ w^{0}$, then the policyholder doesn't purchase insurance until the wealth reaches $ w^{*}$, at which time, it's optimal to buy $n$-year term life insurance of $ f-H^{*} =\frac{rf}{r+\:_{0\shortmid}\textrm{H}_{n}} $.
       	\end{itemize}
       \end{corollary}
       
       \indent When $ m=0$ and $ n\rightarrow \infty $, our problem is equivalent to the problem of maximizing probability of reaching a given bequest goal in \cite{Bayraktar2014},  the concrete results in Corollary \ref{C5} below.
       \begin{corollary}\label{C5}
       	{\rm(1)}. When $ m=0$, $ n\rightarrow \infty $ and if $ \lambda \leqslant r $, then the maximum probability of achieving the financial goal $ f$ before ruining is given by 
       	\begin{equation}
       	_{0\shortmid}\Phi_{\infty}(w)=\displaystyle \biggl ( \frac{w}{H^{*}} \biggr )^{\frac{\lambda }{r}},\notag
       	\end{equation}
       	for initial wealth $ w \in [0,H^{*}) $, in which $H^{*}=\frac{_{0\shortmid}\textrm{H}_{\infty}f}{r+\:_{0\shortmid}\textrm{H}_{\infty}} $.\\
       	\indent The related optimal purchasing strategy is not to purchase until wealth reaches $ H^{*} $, at which point, it's optimal to buy life insurance of $ f-H^{*}=\frac{rf}{r+\:_{0\shortmid}\textrm{H}_{\infty}}$.\\
       	{\rm(2)}. When $ m=0$, $ n\rightarrow \infty $ and if $ \lambda > r$, then the maximum probability of achieving the financial goal $ f$ before ruining is given by 
       	\begin{equation}
       	_{0\shortmid}\Phi_{\infty}(w)=\begin{cases}
       	\displaystyle 1-\biggl(1-\frac{w}{H^{*}}\biggr)^{\frac{\lambda }{r+\:_{0\shortmid}\textrm{H}_{\infty}}},&0 \leqslant w < w^{0},\\[3mm]
       	\displaystyle \biggl ( \frac{w}{H^{*}} \biggr )^{\frac{\lambda }{r}},&w^{0} \leqslant w< H^{*},\\
       	\end{cases}	 \notag
       	\end{equation}
       	for initial wealth $ w \in [0,H^{*}) $, where $ w^{0}$ is the unique zero in $ (0,H^{*})$ of the following equations 
       	\begin{equation}
       	1-\biggl(1-\frac{w}{H^{*}}\biggr)^{\frac{\lambda }{r+\:_{0\shortmid}\textrm{H}_{\infty}}}=\biggl ( \frac{w}{H^{*}} \biggr )^{\frac{\lambda }{r}}. \notag
       	\end{equation}
     The related optimal purchasing strategy is
       	     	\begin{itemize} 
       	\item If wealth $ w $ is less than $ w ^{0}$, then the policyholder purchases life insurance of $ f-w$;
       	\item If wealth $ w$ is greater than or equal to $ w^{0}$, then the policyholder doesn't purchase insurance until the wealth reaches $ H^{*}$, at which point, it's optimal to buy life insurance of $ f-H^{*} =\frac{rf}{r+\:_{0\shortmid}\textrm{H}_{\infty}} $.
       		\end{itemize}
       \end{corollary}
       
       When $ m > 0, n \rightarrow \infty$, the problem turns to consider buying $m$-year deferred whole life insurance, assume $ \frac{_{m\shortmid}\textrm{H}_{\infty}}{r+\:_{m\shortmid}\textrm{H}_{\infty}}<e^{-rm}$, the corresponding results are as follows.
       \begin{corollary}\label{C6}
       	{\rm(1)}. When $m > 0, n \rightarrow \infty$ and if $ \lambda \leqslant r $, then the maximum probability of achieving the financial goal $ f$ before ruining is given by 
       	\begin{equation}
       	_{m\shortmid}\Phi_{\infty}(w)=\begin{cases}
       	\displaystyle \biggl ( \frac{w}{H^{*}} \biggr )^{\frac{\lambda }{r}}e^{-\lambda m},&0 \leqslant w< H^{*},\\[6mm]
       	\displaystyle e^{-\lambda m},&H^{*} \leqslant w< w^{0},\\[1mm]
       	\displaystyle \biggl ( \frac{w-H^{*}}{w^{*}-H^{*}} \biggr )^{\frac{\lambda }{r}},& w^{0} \leqslant w< w^{*},
       	\end{cases}	  \notag
       	\end{equation}
       	in which $ w^{0} = e^{-rm}(w^{*}-H^{*})+H^{*}$, $ H^{*}=\frac{_{m\shortmid}\textrm{H}_{\infty}f}{r+\:_{m\shortmid}\textrm{H}_{\infty}} $,  $w^{*}= \frac{(r+\:_{m\shortmid}\textrm{H}_{\infty}+r\:_{m\shortmid}\textrm{H}_{\infty})f}{(r+\:_{m\shortmid}\textrm{H}_{\infty})(r+1)} $ and the initial wealth $ w \in [0,w^{*}) $.\\
       	\indent The related optimal purchasing strategy is not to purchase until wealth reaches $ w^{*} $, at which point, it's optimal to buy $m$-year deferred whole life insurance of $ f-H^{*} =\frac{rf}{r+\:_{m\shortmid}\textrm{H}_{\infty}}$.\\
       	{\rm(2)}. When $m > 0, n \rightarrow \infty$ and if $ \lambda > r$, then the maximum probability of achieving the financial goal $ f$ before ruining is given by 
       	\begin{equation}
       	_{m\shortmid}\Phi_{\infty}(w)=\begin{cases}
       	\displaystyle \biggl[1-\biggl(1-\frac{w}{H^{*}}\biggr)^{\frac{\lambda }{r+\:_{m\shortmid}\textrm{H}_{\infty}}}\biggr]e^{-\lambda m},&0\leqslant w < w^{0},\\[3mm]
       	\displaystyle \biggl ( \frac{w}{H^{*}} \biggr )^{\frac{\lambda }{r}}e^{-\lambda m},&w^{0} \leqslant w< H^{*},\\[6mm]
       	\displaystyle e^{-\lambda m},&H^{*} \leqslant w< w_{1},\\[1mm]
       	\displaystyle \biggl ( \frac{w-H^{*}}{w^{*}-H^{*}} \biggr )^{\frac{\lambda }{r}},& w_{1} \leqslant w< w^{*},
       	\end{cases}	 \notag
       	\end{equation}
       	in which $w_{1} =e^{-rm}(w^{*}-H^{*})+H^{*}$ and the initial wealth $ w \in [0,w^{*}) $, where $ w^{0}$ is the unique zero in $ (0,H^{*})$ of the following equations 
       	\begin{equation}
       	1-\biggl(1-\frac{w}{H^{*}}\biggr)^{\frac{\lambda }{r+\:_{m\shortmid}\textrm{H}_{\infty}}}= \biggl ( \frac{w}{H^{*}} \biggr )^{\frac{\lambda }{r}}. \notag
       	\end{equation}
   The related optimal purchasing strategy is:
       	\begin{itemize}
       	\item If wealth $ w $ is less than $ w ^{0}$, then the policyholder purchase $m$-year deferred whole life insurance of $ f-w$; 
       	\item If wealth $ w$ is greater than or equal to $ w^{0}$, then the policyholder doesn't purchase insurance until the wealth reaches $ w^{*}$, at which point, it's optimal to buy $m$-year deferred whole life insurance of $ f-H^{*} =\frac{rf}{r+\:_{m\shortmid}\textrm{H}_{\infty}} $.
       		\end{itemize}
       \end{corollary}
       \section{Purchasing $m$-year deferred $n$-year term pure endowment in personal financial planning under the deterministic framework}\label{S3}
       The policyholder purchases $m$-year deferred $n$-year term pure endowment in financial planning to make a longevity risk protection plan. Assume he/she has the financial goal $ f $ to ensure adequate pension at time $ \tau$ which is the retirement moment, and $\tau $ follows an exponential distribution with parameter $ \lambda $. Assumptions and methods are as similar as that in Section \ref{S2}, and we also add the time cutoff $ n$ and $ m$, but because of the policyholder must live for $ m+n$ years, so the discussion about $ n$ , $ m$ and $ \tau$ changes. By the same way as that in previous Section \ref{S2}, we will directly give the relevant conclusions in this section.
       \subsection{$m$-year deferred $n$-year term pure endowment purchased by a single premium}
       The policyholder purchases $m$-year deferred $n$-year term pure endowment by a single premium for $ _{m\shortmid}\textrm{R}_{n}$ per dollar of insurance, in which
       \begin{equation}
       _{m\shortmid}\textrm{R}_{n}=(1+\theta)\int_{m+n}^{\infty}e^{-(m+n)r}\lambda e^{-\lambda t}dt=(1+\theta)e^{-(m+n)(r+\lambda)}. \notag
       \end{equation}
       Similar to Section \ref{S2}, we restrict premium $ _{m\shortmid}\textrm{R}_{n}<e^{-rm}$. The wealth follows the dynamics
       \begin{equation}
       \begin{cases}
       dW(t)=rW(t-)dt-\:_{m\shortmid}\textrm{R}_{n}dD(t),&\ 0\leqslant t<\tau,\\[2mm]W(\tau)=W(\tau-)+D(\tau-)\mathbbm{1}_{\left \{ \tau>m+n \right \}}.
       \end{cases} \notag
       \end{equation}
       We define the maximum probability of achieving the financial goal $ f $ is $ _{m\shortmid}\phi_{n}$ as follows: 
       \begin{equation}
       \begin{split}
       _{m\shortmid}\phi_{n}(w,D)&=\sup_{\textbf{D}}  \biggl[\textbf{P}^{w,D}(W(\tau)\geqslant f \mid \tau \leqslant m+n)P(\tau \leqslant m+n)\\&\quad+ \textbf{P}^{w,D}(W(\tau)\geqslant f \mid \tau>m+n)P(\tau>m+n)  \biggr].		\notag
       \end{split}
       \end{equation}
       \begin{proposition}\label{P3}
       	The maximum probability of achieving the financial goal is given by 
       {\small
       	\begin{equation}
       	_{m\shortmid}\phi_{n}(w,D)=\begin{cases} 
       	 \left ( \frac{w}{R_{*}} \right )^{\frac{ \lambda }{r}}e^{-\lambda(m+n)}, &0\leqslant w< R_{*},\\[3mm]
      e^{-\lambda (m+n)}+\biggl(e^{-\lambda m}-e^{-\lambda (m+n)}\biggr)\left ( \frac{w-R_{*}}{f-D} \right )^{\frac{\lambda }{r}},&R_{*}\leqslant w < w_{0},\\[3mm]
 \biggl(1+e^{-\lambda m}-e^{-\lambda (m+n)}\biggr)\left ( \frac{w-R_{*}}{f-D} \right )^{\frac{\lambda }{r}}-e^{-\lambda m}+e^{-\lambda (m+n)},&w_{0}\leqslant w < w_{*},
       	\end{cases}  \notag
       	\end{equation}
       }
       	in which $R_{*}=\:_{m\shortmid}\textrm{R}_{n}(f-D)$, $w_{0}=(e^{-rm}+\:_{m\shortmid}\textrm{R}_{n})(f-D)$ and $ w_{*}=(_{m\shortmid}\textrm{R}_{n}+1)(f-D) $.\\
       	\indent The related optimal purchasing strategy is not to purchase until wealth reaches $w_{*} $, at which point, it's optimal to buy additional $m$-year deferred $n$-year term pure endowment of $ f-D$.
       \end{proposition}
       \begin{corollary}\label{C7}
       	When $ m=0, n \geqslant 0$, the maximum probability of achieving the financial goal $ f $ on L is given by 
       	\begin{equation}
       	_{0\shortmid}\phi_{n}(w,D)=\begin{cases} 
       	\displaystyle \left ( \frac{w}{R_{*}} \right )^{\frac{ \lambda }{ r}}e^{-\lambda n},&0\leqslant w<R_{*},\\[3mm]
       	\displaystyle e^{-\lambda n}+\biggl(1-e^{-\lambda n}\biggr)\left ( \frac{w-R_{*}}{f-D} \right )^{\frac{\lambda }{r}},&R_{*}\leqslant w<w_{*},
       	\end{cases} \notag
       	\end{equation}
       	in which $R_{*}=\:_{0\shortmid}\textrm{R}_{n}(f-D)$, $ w_{*}=(_{0\shortmid}\textrm{R}_{n}+1)(f-D) $.\\
       	\indent The related optimal purchasing strategy is not to purchase until wealth reaches $w_{*} $, at which point, it's optimal to buy additional $n$-year term pure endowment of $ f-D$.
       \end{corollary}
       \subsection{$m$-year deferred $n$-year term pure endowment purchased by a continuously paid premium}
       We work out the problem about buying instantaneous $m$-year deferred $n$-year term pure endowment via a premium paid continuously at the rate of $ _{m\shortmid}\textrm{M}_{n} $ per dollar of insurance and $\frac{_{m\shortmid}\textrm{M}_{n}+r\:_{m\shortmid}\textrm{M}_{n}}{r+\:_{m\shortmid}\textrm{M}_{n}+r\:_{m\shortmid}\textrm{M}_{n}} <e^{-rm}$, in which
       \begin{equation}
       _{m\shortmid}\textrm{M}_{n}=\frac{(1+\theta )\ax*[m|]{\endow{x}{n}}}{ \ax*{\endow{x}{m}}}. \notag
       \end{equation}
       The wealth satisfies the following dynamics
       \begin{equation}
       \begin{cases}
       dW(t)=(rW(t)-\:_{m\shortmid}\textrm{M}_{n}D(t)\mathbbm{1}_{\left \{t \leqslant m \right \}})dt,&\ 0\leqslant t<\tau,\\[2mm]
       W(\tau)=W(\tau-)+D(\tau-)\mathbbm{1}_{\left \{\tau> m+n \right \}}.
       \end{cases}\notag
       \end{equation}
       The maximum probability of achieving the financial goal $ f$ is as follows: 
       \begin{equation}
       \begin{split}
       _{m\shortmid}\Phi_{n}(w)&=\sup_{\textbf{D}}  \biggl[\textbf{P}^{w}(W(\tau \wedge \delta _{0})\geqslant f \mid \tau \leqslant m+n)P(\tau \leqslant m+n)\\&\quad+ \textbf{P}^{w}(W(\tau \wedge \delta _{0})\geqslant f \mid \tau>m+n)P(\tau>m+n)  \biggr].		\notag
       \end{split}
       \end{equation}
       \begin{proposition}\label{P4}
       	{\rm(1)}. If $ \lambda \leqslant r $, then the maximum probability of achieving the financial goal $ f$ before ruining is given by 
       	{\small
       	\begin{equation}
       	_{m\shortmid}\Phi_{n}(w)=\begin{cases}
       	\displaystyle \biggl ( \frac{w}{M^{*}} \biggr )^{\frac{\lambda }{r}}e^{-\lambda (m+n)},&0 \leqslant w< M^{*},\\[3mm]
       	\displaystyle e^{-\lambda (m+n)}+\biggl(e^{-\lambda m}-e^{-\lambda (m+n)}\biggr)\biggl ( \frac{w-M^{*}}{w^{*}-M^{*}} \biggr )^{\frac{\lambda }{r}},&M^{*} \leqslant w< w^{0},\\[3mm]
       	\displaystyle \biggl(1+e^{-\lambda m}-e^{-\lambda (m+n)}\biggr)\biggl ( \frac{w-M^{*}}{w^{*}-M^{*}} \biggr )^{\frac{\lambda }{r}}-e^{-\lambda m}+e^{-\lambda (m+n)},& w^{0} \leqslant w< w^{*},
       	\end{cases}	  \notag
       	\end{equation}
       }
       	in which $M^{*}= \frac{_{m\shortmid}\textrm{M}_{n}f}{r+\:_{m\shortmid}\textrm{M}_{n}}$, $ w^{0} = e^{-rm}(w^{*}-M^{*})+M^{*}$, $ w^{*}= \frac{(r+\:_{m\shortmid}\textrm{M}_{n}+r\:_{m\shortmid}\textrm{M}_{n})f}{(r+\:_{m\shortmid}\textrm{M}_{n})(r+1)}$ and the initial wealth $ w \in [0,w^{*}) $.\\
       	\indent The related optimal purchasing strategy is not to purchase until wealth reaches $ w^{*} $, at which point, it's optimal to buy $m$-year deferred $n$-year term pure endowment of $ f-M^{*} =\frac{rf}{r+\:_{m\shortmid}\textrm{M}_{n}}$.\\
       	{\rm(2)}. If $ \lambda > r$, then the maximum probability of achieving the financial goal $ f$ before ruining is given by 
       	{\small
       	\begin{equation}
       	_{m\shortmid}\Phi_{n}(w)=\begin{cases}
       	\displaystyle \biggl[1-\biggl(1-\frac{w}{M^{*}}\biggr)^{\frac{\lambda }{r+\:_{m\shortmid}\textrm{M}_{n}}}\biggr]e^{-\lambda (m+n)},&0 \leqslant w < w^{0},\\[3mm]
       	\displaystyle \biggl ( \frac{w}{M^{*}} \biggr )^{\frac{\lambda }{r}}e^{-\lambda (m+n)},&w^{0} \leqslant w< M^{*},\\[3mm]
       	\displaystyle e^{-\lambda (m+n)}+\biggl(e^{-\lambda m}-e^{-\lambda (m+n)}\biggr)\biggl ( \frac{w-M^{*}}{w^{*}-M^{*}} \biggr )^{\frac{\lambda }{r}},&M^{*} \leqslant w< w_{1},\\[3mm]
       	\displaystyle \biggl(1+e^{-\lambda m}-e^{-\lambda (m+n)}\biggr)\biggl ( \frac{w-M^{*}}{w^{*}-M^{*}} \biggr )^{\frac{\lambda }{r}}-e^{-\lambda m}+e^{-\lambda (m+n)},& w_{1} \leqslant w< w^{*},
       	\end{cases}	 \notag
       	\end{equation}
       }
       	in which $w_{1} = e^{-rm}(w^{*}-M^{*})+M^{*}$, the initial wealth $w \in [0,w^{*})$, where $ w^{0}$ is the unique zero in $ (0,M^{*})$ of the following equations 
       	\begin{equation}
       	1-\biggl(1-\frac{w}{M^{*}}\biggr)^{\frac{\lambda }{r+\:_{m\shortmid}\textrm{M}_{n}}}= \biggl ( \frac{w}{M^{*}} \biggr )^{\frac{\lambda }{r}}. \notag
       	\end{equation}
     The related optimal purchasing strategy is
       	\begin{itemize}
       	\item If wealth $ w $ is less than $ w ^{0}$, then the policyholder purchases $m$-year deferred $n$-year term pure endowment of $ f-w$;
       	\item If wealth $ w$ is greater than or equal to $ w^{0}$, then the policyholder doesn't purchase until the wealth reaches $ w^{*}$, at which point, it's optimal to buy $m$-year deferred $n$-year term pure endowment of $ f-M^{*} =\frac{rf}{r+\:_{m\shortmid}\textrm{M}_{n}} $.
       		\end{itemize}
       \end{proposition}
       \begin{corollary}\label{C8}
       	{\rm(1)}. When $ m=0 , n \geqslant 0$ and if $ \lambda \leqslant r $, then the maximum probability of achieving the financial goal $ f $ before ruining is given by 
       	\begin{equation}
       	_{0\shortmid}\Phi_{n}(w)=\begin{cases}
       	\displaystyle \biggl ( \frac{w}{M^{*}} \biggr )^{\frac{\lambda }{r}}e^{-\lambda n},&0 \leqslant w< M^{*},\\[3mm]
       	\displaystyle e^{-\lambda n}+\biggl(1-e^{-\lambda n}\biggr)\left ( \frac{w-M^{*}}{w^{*}-M^{*}} \right )^{\frac{\lambda }{r}},&M^{*} \leqslant w< w^{*},
       	\end{cases}	\notag
       	\end{equation}
       	in which $M^{*}= \frac{_{0\shortmid}\textrm{M}_{n}f}{r+\:_{0\shortmid}\textrm{M}_{n}}$,  $w^{*}=\frac{(r+\:_{0\shortmid}\textrm{M}_{n}+r\:_{0\shortmid}\textrm{M}_{n})f}{(r+\:_{0\shortmid}\textrm{M}_{n})(r+1)} $, for initial wealth $ w \in [0,w^{*}) $.\\
       	\indent The related optimal $n$-year term pure endowment purchasing strategy is not to purchase until wealth reaches $ w^{*} $, at which point, it's optimal to buy $n$-year term pure endowment of $ f-M^{*} =\frac{rf}{r+\:_{0\shortmid}\textrm{M}_{n}}$.\\
       	{\rm(2)}. When $ m=0,  n \geqslant 0$ and if $ \lambda > r$, then the maximum probability of achieving the financial goal $ f $ before ruining is given by 
       	\begin{equation}
       	_{0\shortmid}\Phi_{n}(w)=\begin{cases}
       	\displaystyle \biggl[1-\biggl(1-\frac{w}{M^{*}}\biggr)^{\frac{\lambda }{r+\:_{0\shortmid}\textrm{M}_{n}}}\biggr]e^{-\lambda n},&0 \leqslant w < w^{0},\\[3mm]
       	\displaystyle \biggl ( \frac{w}{M^{*}} \biggr )^{\frac{\lambda }{r}}e^{-\lambda n},&w^{0} \leqslant w< M^{*},\\[3mm]
       	\displaystyle e^{-\lambda n}+\biggl(1-e^{-\lambda n}\biggr)\left ( \frac{w-M^{*}}{w^{*}-M^{*}} \right )^{\frac{\lambda }{r}},&M^{*} \leqslant w< w^{*},
       	\end{cases}	\notag
       	\end{equation}
       	for initial wealth $ w \in [0,w^{*}) $, where $ w^{0}$ is the unique zero in $ (0,M^{*})$ of the following equations 
       	\begin{equation}
       	1-\biggl(1-\frac{w}{M^{*}}\biggr)^{\frac{\lambda }{r+\:_{0\shortmid}\textrm{M}_{n}}}=\biggl ( \frac{w}{M^{*}} \biggr )^{\frac{\lambda }{r}}. \notag
       	\end{equation}
    The related optimal purchasing strategy is 
       	\begin{itemize}
       \item If wealth $ w $ is less than $ w ^{0}$, then the policyholder purchases $n$-year term pure endowment of $ f-w$;
       	\item If wealth $ w$ is greater than or equal to $ w^{0}$, then the policyholder doesn't purchase until the wealth reaches $ w^{*}$, at which point, it's optimal to buy $n$-year term pure endowment of $ f-M^{*} =\frac{rf}{r+\:_{0\shortmid}\textrm{M}_{n}} $.
       \end{itemize}
       \end{corollary}
       
       \section{Purchasing  $n$-year term life insurance in personal financial planning under the stochastic framework}\label{S4}
       Assume the policyholder has an account to achieve a financial goal, and this account include the consuming, incoming and risky investment. Specifically, he/she invest in a financial market which including two parts, one is investing risk-free asset to earn interest, we denote the force of interest is $r$, another is investing in a risky market whose process $ S=\left\{S(t)\right\}_{t \geqslant 0}$ follows geometric Brownian motion. The individual buys $n$-year term life insurance via a premium paid continuously at the rate of $H_{n}$ per dollar of insurance. In this section, the definitions of $\delta_{d},\delta_{0},W(t), D(t) $ are same as the previous sections.
       \subsection{$n$-year term life insurance purchased by a continuously paid premium under the model \uppercase\expandafter{\romannumeral1}}
       \indent  The specific components of model \uppercase\expandafter{\romannumeral1} are as follows
       \begin{equation}
       \begin{cases}
       dS(t)=\mu S(t)dt+\sigma S(t)dB(t),\\
       dY(t)=aW(t)dt+ldB(t),\\
       c(W)=cW(t),\notag
       \end{cases} 
       \end{equation}
       in which $ Y=\left\{Y(t)\right\}_{t\geqslant0}$ represents the price process of income, $ c=\left\{cW(t)\right\}_{t\geqslant0}$ represents the consumption rate. $ B=\left\{B(t)\right\}_{t\geqslant0}$ is a standard Brownian motion on a filtered probability space $ (\Omega ,\mathcal{F}, \mathbb{F}=\left\{\mathcal{F}_{t}\right\}_{t\geqslant 0},\mathbb{P})$, $\mu $, $ \sigma$, $ a$, $ l$, $ c$ are all constants, in particular, we assume $ \mu >r$ and $ c<r+a$.
       
       \indent Denote $ \pi_{t} $ is the amount invested in the risky market at time $ t\geqslant0$. An investment strategy $\Pi=\left\{ \pi _{t}\right\}_{t\geqslant 0}  $ is admissible if it is an $ \mathcal{F}$-progressively measurable process satisfying $ \int_{0}^{t}\pi _{s}^{2}ds<\infty $. An admissible insurance strategy $ \textbf{D}=\left \{ D(t) \right \}_{t\geqslant 0}$ is any non-negative and $ \mathcal{F}$-progressively measurable process, but for all $ t \geqslant0$, the probability of $ W(t) \geqslant 0$ doesn't equal one because of the negative drift term $ H_{n}D(t)$. Therefore, the wealth follows the dynamics
       {\small
       \begin{equation}
       \begin{cases}
       dW(t)=\bigl [ (r+a-c)W(t)+(\mu -r)\pi _{t}-H_{n}D(t)\mathbbm{1}_{\left \{ t\leqslant n \right \}} \bigr ]dt+(\sigma \pi _{t}+l)dB_{t},&\ 0\leqslant t<\delta _{d},\\W(\delta _{d})=W(\delta _{d}-)+D(\delta _{d}-)\mathbbm{1}_{\left \{ \delta _{d}\leqslant n \right \}}.
       \end{cases} \notag
       \end{equation}
    }
       \indent The maximum probability of achieving the financial goal $f$  
       {\small 
       \begin{equation}
       \begin{split}
       \Phi_{n}(w)&=\sup_{\textbf{D}}  \biggl[\textbf{P}^{w}(W(\delta _{d} \wedge \delta _{0}) \geqslant f \mid  \delta_{d} \leqslant n)P( \delta_{d} \leqslant n)\\&\quad +\textbf{P}^{w}(W(\delta _{d} \wedge \delta _{0}) \geqslant f \mid \delta_{d} > n)P(\delta_{d} > n)  \biggr]\\
       &:=\sup_{\textbf{D}}  \bigl(\textbf{P}^{w}_{1}+ \textbf{P}^{w}_{2} \bigr),
       \end{split}		\notag
       \end{equation}    
    }
       in which $\textbf{P}^{w} $ denotes conditional probability given $ W(0)=w\geqslant0$.
       \begin{remark}
       	{\rm(1)}. Firstly, we define``quasi-ideal value" and``ideal value". If wealth equals $ w^{q_{0}}=\frac{H_{n}f}{r+a-c+H_{n}}$ which is solved from the equation $ (r+a-c)w^{q_{0}}=H_{n}(f-w^{q_{0}}) $, then it's called the ``quasi-ideal value". It means that if wealth reaches  $ \frac{H_{n}f}{r+a-c+H_{n}}$, then the policyholder purchases $n$-year term life insurance of $ f- \frac{H_{n}f}{r+a-c+H_{n}} $ via a premium paid continuously, and if he/she dies within $n$ years after purchasing insurance, then the individual's total death benefit becomes $ f$.  If the policyholder dies after $n$ years of purchasing insurance, he/she may not receive the death benefit, so in this case, he/she can't achieve the financial goal $ f$. Thus, we call$\frac{H_{n}f}{r+a-c+H_{n}} $ is the ``quasi-ideal value".\\
       	{\rm(2)}. Supposing the ``ideal value" is $ w^{i_{0}} $, if wealth equals $ w^{i_{0}} $, then it's optimal for the policyholder to purchase $n$-year term life insurance of $ f- \frac{H_{n}f}{r+a-c+H_{n}} $, then whether or not he/she can get the death benefit, he/she will achieve the financial goal $ f$. Deriving from our setting, we obtain $ w^{i_{0}} $ by following equation 
       	\begin{equation}
       	(r+a-c)w^{i_{0}}-H_{n}\biggl(f-\frac{H_{n}f}{r+a-c+H_{n}}\biggr)=f-w^{i_{0}} ,\notag
       	\end{equation}
       	thus, we get $ w^{i_{0}}=\displaystyle \frac{\bigl(r+a-c+H_{n}+(r+a-c)H_{n}\bigr)f}{(r+a-c+H_{n})(r+a-c+1)}$.
       \end{remark}
       
       To motivate the verification lemma for this problem, we first denote $ w_{b_{0}}$ which is called the buy level, and it's value will be given in subsequent discussions. We give the control equation about $ \Phi_{n} $ as follows: 
       {\small
       \begin{align}
       \lambda \Bigl[ \Phi_{n} -&(1-e^{-\lambda n})\mathbbm{1}_{\left\{ w\geqslant w_{b_{0}}\right\}} \Bigr]
       =\Bigl[ (r+a-c)w-H_{n}\Bigl(f-w\mathbbm{1}_{\left\{ w<\frac{H_{n}f}{r+a-c+H_{n}}+\frac{(\mu -r)l}{\sigma (r+a-c+H_{n})}\right\}}\notag\\-&\frac{H_{n}f}{r+a-c+H_{n}}\mathbbm{1}_{\left\{ w\geqslant \frac{H_{n}f}{r+a-c+H_{n}}+\frac{(\mu -r)l}{\sigma (r+a-c)} \right\}}\Bigr)\mathbbm{1}_{\left\{ w_{b_{0}}\leqslant w\leqslant w^{i_{0}}\right\}} \Bigr](\Phi_{n}) _{w}\notag\\+&\max_{\pi }\Bigl[ (\mu -r)\pi (\Phi_{n}) _{w}+\frac{1}{2}(\sigma \pi +l)^{2}(\Phi_{n}) _{ww} \Bigr] \label{e2}
       \end{align}
    }
       \textbf{Step1}. We simplify the last term in above control equation. 
       \begin{align}\label{e3}
       &\max_{\pi }\Bigl[ (\mu -r)\pi (\Phi_{n}) _{w}+\frac{1}{2}(\sigma \pi +l)^{2}(\Phi_{n})_{ww} \Bigr]\\=&
       \max_{\pi }\Bigl[ (\mu -r)\pi(\Phi_{n}) _{w}+\frac{1}{2}(\sigma^{2} \pi^{2} +l^{2}+2\sigma l \pi)(\Phi_{n}) _{ww} \Bigr]\notag\\
       =&\max_{\pi }\Bigl[\frac{1}{2}\sigma^{2}(\Phi_{n}) _{ww}\pi ^{2}+[(\mu -r)(\Phi_{n})_{w}+\sigma l(\Phi_{n})_{ww}]\pi +\frac{1}{2}l^{2}(\Phi_{n}) _{ww} \Bigr]\notag\\
       =&-\frac{1}{2}\Bigl(\frac{\mu -r}{\sigma }\Bigr)^{2}\frac{(\Phi_{n})_{w}^{2}}{(\Phi_{n}) _{ww}}-\frac{(\mu -r)l}{\sigma }(\Phi_{n})_{w}\notag
       \end{align}
       \textbf{Step2}. We simplify the above control equation into three parts.\\
       \noindent (\textbf{\romannumeral1}). If $ w<w_{b_{0}}$ (it can be seen from the solution of the latter equation that in fact $\frac{(\mu -r)l}{\sigma (r+a-c)} \leqslant w <w_{b_{0}} $), then the control equation can be replaced with the following equivalent expression
       \begin{equation}
       \lambda\Phi_{n}=\Bigl[(r+a-c)w-\frac{(\mu -r)l}{\sigma }\Bigr](\Phi_{n}) _{w}-\frac{1}{2}\Bigl(\frac{\mu -r}{\sigma }\Bigr)^{2}\frac{(\Phi_{n})_{w}^{2}}{(\Phi_{n})_{ww}}. \label{e4}
       \end{equation}   
       \noindent (\textbf{\romannumeral2}). If $ w_{b_{0}} \leqslant w<\frac{H_{n}f}{r+a-c+H_{n}}+\frac{(\mu -r)l}{\sigma (r+a-c+H_{n})}$, then the control equation can be replaced with the following equivalent expression
       {\small 
       \begin{equation}
       \lambda\Bigl( \Phi_{n} -(1-e^{-\lambda n})\Bigr)=\Bigl[(r+a-c+H_{n})w-(H_{n}f+\frac{(\mu -r)l}{\sigma })\Bigr](\Phi_{n}) _{w}-\frac{1}{2}\Bigl(\frac{\mu -r}{\sigma }\Bigr)^{2}\frac{(\Phi_{n}) _{w}^{2}}{(\Phi_{n})_{ww}}. \label{e5}
       \end{equation}
    }
       \noindent (\textbf{\romannumeral3}). If $ \frac{H_{n}f}{r+a-c+H_{n}}+\frac{(\mu -r)l}{\sigma (r+a-c)} \leqslant w<w^{i_{0}}$, then the control equation can be replaced with the following equivalent expression
       {\small 
       \begin{equation}
       \lambda\Bigl( \Phi_{n} -(1-e^{-\lambda n})\Bigr)=\Bigl[(r+a-c)w-(\frac{(r+a-c)H_{n}f}{r+a-c+H_{n}}+\frac{(\mu -r)l}{\sigma })\Bigr](\Phi_{n}) _{w}-\frac{1}{2}\Bigl(\frac{\mu -r}{\sigma }\Bigr)^{2}\frac{(\Phi_{n})_{w}^{2}}{(\Phi_{n})_{ww}}. \label{e6}
       \end{equation}
    }
       These observations lead to the following verification Lemma \ref{l1}. 
       \begin{lemma}\label{l1}
       	Define a differential operator  $ \mathscr{L}_{1}^{\pi , D}$: 
       	{\small 
       	\begin{equation}
       	\mathscr{L}_{1}^{\pi,D}\tilde{F}=\bigl[(r+a-c)w+(\mu-r)\pi-H_{n}D\bigr]\tilde{F} _{w}+\frac{1}{2}(\sigma \pi +l)^{2}\tilde{F} _{ww}-\lambda \bigl[\tilde{F}-\lambda(1-e^{-\lambda n}) \mathbbm{1}_{\left \{ w+D\geqslant f \right \}}\bigr]. \notag
       	\end{equation}
       }
       	Let $ F=F(w)$ be a function that is non-decreasing, continuous, and piecewise differentiable on $ [0,w^{i_{0}}) $, except that $ F $ might not be differentiable at $ 0$, $ w_{0}$, $ w_{1}$ and $ w_{2}$. Suppose $ F$ satisfies $ F _{w} >0$, $ F _{ww}<0$ on $ (0,w^{i_{0}})\setminus \bigl((0,w_{0}) \cup (w_{1},w_{2}) \bigr)$. If $ F$ satisfies the following boundary-value problem on $ [0,w^{i_{0}}]$:
       	\begin{equation}
       	\begin{cases}
       	\max_{\pi,D\geqslant0 } \mathscr{L}_{1}^{\pi,D}F(w)=0 ,\\
       	F(0)=F(w_{0})=0, \quad F(w_{1})=F(w_{2})=1-e^{-\lambda n}, \quad F(w^{i_{0}})=1.
       	\end{cases}	 \notag
       	\end{equation}
       	Then, on $ [0,w^{i_{0}}]$, 
       	\begin{equation}
       	\Phi_{n}=F. \notag
       	\end{equation}
       	The optimal investment amount in risk market is 
       	\begin{equation}
       	\pi ^{*}_{t}=-\frac{\mu -r}{\sigma }^{2}\frac{(\Phi_{n}) _{w}(w^{*}_{t})}{(\Phi_{n})_{ww}(w^{*}_{t})}, \notag
       	\end{equation}
       	in which $ w^{*}_{t} $ is the optimal control wealth value at time $ t$.
       	The optimal $n$-year term life insurance purchase amount is
       	\begin{equation}
       	D ^{*}_{t}=\Bigl(f-w\mathbbm{1}_{\left\{ w<w_{1 }\right\}}-\frac{H_{n}f}{r+a-c+H_{n}}\mathbbm{1}_{\left\{ w\geqslant w_{2} \right\}}\Bigr)\mathbbm{1}_{\left\{ w_{b_{0}}\leqslant w\leqslant w^{i_{0}}\right\}} , \notag
       	\end{equation}
       	in which 
       	{\footnotesize
       	\begin{equation}
       	w_{b_{0}}=\rm{inf}\left\{w\geqslant 0: \lambda(1-e^{-\lambda n})-H_{n}\Bigl(f-w\mathbbm{1}_{\left\{ w<w_{1 }\right\}}-\frac{H_{n}f}{r+a-c+H_{n}}\mathbbm{1}_{\left\{ w\geqslant w_{2} \right\}}\Bigr)F_{w}(w) \geqslant0 \right\} \wedge f, \notag
       	\end{equation}
       }
   {\small
       	\begin{equation}
       	w_{0}=\frac{(\mu -r)l}{\sigma (r+a-c)}, w_{1}=\frac{H_{n}f}{r+a-c+H_{n}}+\frac{(\mu -r)l}{\sigma (r+a-c+H_{n})}, \notag
       	\end{equation}
       	\begin{equation}
       	w_{2}=\frac{H_{n}f}{r+a-c+H_{n}}+\frac{(\mu -r)l}{\sigma (r+a-c)}, \notag 
       	\end{equation}
       }
       \end{lemma}
       \begin{proof}
       	\rm Rewrite the express for $ \Phi_{n}$ as follows: 
       	{\small 
       	\begin{equation}
       	\Phi_{n} (w)=\sup_{\pi , \textbf{D}}\textbf{E}^{w}\left [ \int_{0}^{\delta _{0}}\lambda e^{-\lambda t}(1-e^{-\lambda n})\mathbbm{1}_{\left \{W(t)+D(t)\geqslant f\right \}}dt \right ] . \notag
       	\end{equation} 
       }
       \indent 	Let F satisfies the conditions in the Lemma, firstly prove that the conclusions hold under the additional conditions: (1). $F$ has a lower bound, i.e. $ F \geqslant V>- \infty$; (2). $ F_{w}(0)<+\infty$.\\
      \indent Denote $ \delta^{a}_{n}:=\rm{inf} \left\{ s>0:\int_{0}^{s}\pi ^{2}_{t}dt \geqslant n\right\} $, let $ \delta_{n}=\delta _{0} \wedge \delta^{a}_{n}$, then applying the It\^o's formula for $ e^{-\lambda \tau _{n}}F(W_{\tau _{n}})$, we have
       	\begin{align}
       	e^{-\lambda \tau _{n}}F(W_{\tau _{n}})=F(w)+&\int_{0}^{\delta _{n}}e^{-\lambda t}F_{w}(W_{t})(\sigma \pi _{t}+l)dB_{t} \notag \\+&
       	\int_{0}^{\delta _{n}}e^{-\lambda t}[\mathscr{L}_{1}^{\pi,D}F(W_{t})-\lambda (1-e^{-\lambda n})\mathbbm{1}_{\left\{ W_{t}+D_{t}\geqslant f\right\}}]dt \notag \\+&
       	\sum_{0\leqslant t\leqslant \delta _{n}}^{}e^{-\lambda t}[F(W_{t})-F(W_{t-})]. \notag
       	\end{align}
       	Due to
       	\begin{equation}
       	F^{2}_{w}(w)\leqslant F^{2}_{w}(0),  w\geqslant 0 , \quad \textbf{E}^{w}\left [ \int_{0}^{\delta _{n}}e^{-2\lambda t}F^{2}_{w}(W_{t})(\sigma \pi _{t}+l)^{2}dt \right ]< \infty , \notag
       	\end{equation}
       	therefore
       	\begin{equation}
       	\textbf{E}^{w}\left [ \int_{0}^{\delta _{n}}e^{-\lambda t}F_{w}(W_{t})(\sigma \pi _{t}+l)dB_{t}\right ]=0. \notag
       	\end{equation}
       	Then we have 
       	\begin{equation}
       	\textbf{E}^{w}[e^{-\lambda \delta _{n}}V]\leqslant \textbf{E}^{w}[e^{-\lambda \delta _{n}}F(W_{\delta _{n}})]\leqslant 
       	F(w)-\textbf{E}^{w}[\int_{0}^{\delta _{n}}e^{-\lambda t}\lambda (1-e^{-\lambda n})\mathbbm{1}_{\left\{ W_{t}+D_{t}\geqslant f\right\}}dt]. \notag
       	\end{equation}
       	When $ n\rightarrow \infty $, then $ \delta _{n}\rightarrow \infty$, naturally 
       	\begin{equation}
       	F(w) \geqslant \textbf{E}^{w}\left [ \int_{0}^{\delta _{0}}\lambda e^{-\lambda t}(1-e^{-\lambda n})\mathbbm{1}_{\left \{W(t)+D(t)\geqslant f\right \}}dt \right ], \notag
       	\end{equation}
       	that means $ F \geqslant \Phi_{n}$.\\
       \indent	Next we prove that the conclusion $ F \geqslant \Phi_{n}$ still holds when the above additional conditions (1) and (2) are removed.
       	Suppose $ \varepsilon _{n}$ monotonically decreases to $ 0$, denote $F^{\varepsilon _{n}}(w)=F(w+\varepsilon _{n}) $, then
       	\begin{equation}
       	0 \geqslant \mathscr{L}_{1}^{\pi,D}F(w+\varepsilon _{n})=\mathscr{L}_{1}^{\pi,D}F^{\varepsilon _{n}}(w)+(r+a-c+1)\varepsilon _{n}F^{\varepsilon _{n}}_{w}(w), \notag
       	\end{equation}
       	due to $ (r+a-c+1)\varepsilon _{n}F^{\varepsilon _{n}}_{w}(w) \geqslant 0$, then $ =\mathscr{L}_{1}^{\pi,D}F^{\varepsilon _{n}}(w) \leqslant 0$, $ F^{\varepsilon _{n}}$ satisfied the conditions in Lemma, repeat the above steps to get $ F^{\varepsilon _{n}} \geqslant \Phi_{n}$.
       	From $ F(w)=\displaystyle \lim_{n \to \infty }F^{\varepsilon _{n}}(w) \geqslant \Phi_{n}(w)$, we can obtain 
       	\begin{equation}
       	F \geqslant \Phi_{n}. \notag
       	\end{equation}
       \indent 	By \cite{WY2012}, for $ F$ satisfies the boundary conditions, then $F = \Phi_{n} $.
       	Then, $\pi ^{*}_{t}$ can be obtained from Step1, and $ D^{*}_{t}$ naturally be given by the control equation
       	(\ref{e2}).
        \end{proof}
       \textbf{Step3}. We separately solve the above three control equations. Since the three equations are basically the same in form, we take the first equation as an example for a specific solution.
       \begin{lemma}\label{l2}
       	The solution to equation (\ref{e4}) is of the form as follows
       	\begin{equation}\label{e7}
       	\Phi_{n}(w)=D_{1}\Bigl(w-\frac{(\mu -r)l}{\sigma (r+a-c)}\Bigr)^{p}
       	\end{equation}
       	in which $ D_{1}$ is a constant to be determined and $ p$ is a known constant.
       \end{lemma} 
       \begin{proof}
       	\rm Let $  m= \frac{1}{2}\Bigl(\frac{\mu -r}{\sigma }\Bigr)^{2}$, $ A=r+a-c$, $ B= \frac{\mu -r}{\sigma }$. Then
       	\begin{equation}
       	\lambda \Phi_{n}=(Aw-B)(\Phi_{n})_{w}-m\frac{(\Phi_{n}) _{w}^{2}}{(\Phi_{n}) _{ww}}. \notag
       	\end{equation} 
       	\indent We consider the Legendre transform: 
       	\begin{equation}
       	\varphi(y)=\sup_{w}(\Phi_{n}(w)-wy) .\notag
       	\end{equation}
       	Then, 
       	\begin{equation}
       	\Phi_{n}(w)=\inf_{y}(\varphi (y)+wy)=\varphi (y)-y \varphi '(y), \notag
       	\end{equation}
       	\begin{equation}
       	\varphi '(y)=-w,\quad \varphi ''(y)=-\frac{1}{\Phi_{n}''(w)}.\notag
       	\end{equation}
       	Therefore, 
       	\begin{equation}
       	\lambda [\varphi(y)--y \varphi '(y)]=(-A\varphi '(y)-B)y+my^{2} \varphi ''(y). \notag
       	\end{equation}
       	Equivalently, we obtain
       	\begin{equation}\label{e8}
       	my^{2} \varphi ''(y)+(\lambda-A)y\varphi '(y)-\lambda \varphi(y)-By=0.
       	\end{equation}
       \indent	Firstly, we solve the Euler equation as follows
       	\begin{equation}
       	y^{2} \varphi ''(y)+\frac{(\lambda-A)}{m}y\varphi '(y)-\frac{\lambda}{m} \varphi(y)=0. \notag
       	\end{equation}
       	The corresponding characteristic equation is
       	\begin{equation}
       	x^{2}+\frac{\lambda -A-m}{m}x-\frac{\lambda }{m}=0.\notag
       	\end{equation}
       	Solve the above characteristic equation to get the characteristic root as follows
       	\begin{equation}
       	x_{1}=\frac{A-\lambda +m-\sqrt{(A-\lambda +m)^{2}+4\lambda m}}{2m}<0, \notag
       	\end{equation}
       	\begin{equation}
       	x_{2}=\frac{A-\lambda +m+\sqrt{(A-\lambda +m)^{2}+4\lambda m}}{2m}>1  , \notag	
       	\end{equation}
       	So the general solution of this Euler equation is
       	\begin{equation}
       	\tilde{\varphi }(y)=D_{1}y^{x_{1}}+D_{2}y^{x_{2}}, \notag
       	\end{equation}
       	in which $ D_{1}$ is a constant and $ D_{2}=0$ because $ \varphi '(y) \leqslant 0$ holds for any $ y$.\\
       	A particular solution to the original equation (\ref{e8}) is obtained as follows
       	\begin{equation}
       	\varphi^{*}(y)=\tilde{D} y^{x_{1}}-\frac{B}{A}y .\notag
       	\end{equation}
       	in which $ \tilde{D}$ is a constant.\\
       	So we obtain the solution of original equation (\ref{e8})
       	\begin{equation}
       	\varphi(y)=2\tilde{D_{1}} y^{x_{1}}-\frac{B}{A}y .\notag
       	\end{equation}
       	in which $ \tilde{D_{1}} $ is a constant.\\
       	\indent Due to the Legendre transform, we have
       	\begin{equation}
       	\Phi_{n}(w)=2\tilde{D_{1}} (\Phi_{n}'(w))^{x_{1}}-\frac{B}{A}\Phi_{n}'(w)+w\Phi_{n}'(w). \notag
       	\end{equation}
       	Let $ p= \frac{x_{1}}{x_{1}-1} \in (0,1)$, the solution $ \Phi_{n}(w)$ should be of the form as follows
       	\begin{equation}
       	\Phi_{n}(w)=D_{1}\Bigl(w-\frac{B}{A}\Bigr)^{p}. \notag
       	\end{equation}
       	Substitute the value of $ A$ and $ B$, we have 
       	\begin{equation}
       	\Phi_{n}(w)=D_{1}\Bigl(w-\frac{(\mu -r)l}{\sigma (r+a-c)}\Bigr)^{p}. \notag
       	\end{equation}
       	in which $ D_{1}$ is a constant to be determined. 
       \end{proof}
       Similarly, we have the Lemma \ref{l3} and Lemma \ref{l4}.
       \begin{lemma}\label{l3}
       	The solution to equation (\ref{e5}) is of the form as follows
       	\begin{equation}\label{e9}
       	\Phi_{n}(w)=1-e^{-\lambda n}+D_{2}\Bigl(w-\frac{(\mu -r)l+\sigma H_{n}f}{\sigma (r+a-c+H_{n})}\Bigr)^{q}
       	\end{equation}
       	in which $ D_{1}$ is a constant to be determined, $ q=\frac{k_{1}}{k_{1}-1}>1$, 
       	\begin{equation}
       	k_{1}=\frac{A_{1}-\lambda -m+\sqrt{(A_{1}-\lambda +m)^{2}+4\lambda m}}{2m}, \quad A_{1}=r+a-c+H_{n}\notag
       	\end{equation}
       \end{lemma}
       \begin{lemma}\label{l4}
       	The solution to equation (\ref{e6}) is of the form as follows
       	\begin{equation}\label{e10}
       	\Phi_{n}(w)=1-e^{-\lambda n}+D_{3}\Biggl(w-\Bigl(\frac{H_{n}f}{r+a-c+H_{n}}+\frac{(\mu -r)l}{\sigma (r+a-c)}\Bigr)\Biggr)^{p}
       	\end{equation}
       	in which $ D_{3}$ is a constant to be determined, the value of $ p$ is equal to that in the Lemma \ref{l2}.
       \end{lemma}
   
      \indent  $ D_{1}$, $ D_{2}$, $ D_{3}$ can be determined according to continuity and critical value, the specific method for their determination will be presented in the appendix. We will directly give the concrete values of this parameters in the following propositions.
       \begin{proposition}\label{p1}
       	The maximum probability of achieving the financial goal $ f$ before ruining is given by 
       	\begin{equation}
       	\Phi_{n}(w)=\begin{cases}
       	\displaystyle 0,&0 \leqslant w <w_{0} ,\\
       	\displaystyle (1-e^{-\lambda n})\frac{q(1-p)}{q-p}\left ( \frac{w-w_{0} }{w_{b_{0}}-w_{0}} \right )^{p} ,&w_{0}\leqslant w < w_{b_{0}},\\
       	\displaystyle (1-e^{-\lambda n})\left ( 1-\frac{p(q-1)}{q-p}\Biggl(\frac{w-w_{1}}{w_{b_{0}}-w_{1}}\Biggr)^{q} \right ),&w_{b_{0}} \leqslant w< w_{1},\\
       	\displaystyle 1-e^{-\lambda n},&w_{1} \leqslant w<w_{2}  ,\\
       	\displaystyle 1-e^{-\lambda n}+e^{-\lambda n} \left ( \frac{w- w_{2}}{w^{i_{0}}-w_{2}} \right )^{p},& w_{2}  \leqslant w< w^{i_{0}},
       	\end{cases}	 \notag
       	\end{equation}
       	in which 
   {\footnotesize
       	\begin{equation}
       	w_{0}=\frac{(\mu -r)l}{\sigma (r+a-c)}, w_{1}=\frac{H_{n}f}{r+a-c+H_{n}}+\frac{(\mu -r)l}{\sigma (r+a-c+H_{n})},
       	w_{2}=\frac{H_{n}f}{r+a-c+H_{n}}+\frac{(\mu -r)l}{\sigma (r+a-c)}, \notag 
       	\end{equation}
       }
      	{\small
       	\begin{equation}
       	p= \frac{k_{1}}{k_{1}-1}, \quad k_{1}=\frac{r+a-c-\lambda +m-\sqrt{(r+a-c-\lambda +m)^{2}+4\lambda m}}{2m} \notag
       	\end{equation}
       	\begin{equation}
       	q= \frac{k_{2}}{k_{2}-1}, \quad k_{2}=\frac{r+a-c+H_{n}-\lambda +m+\sqrt{(r+a-c+H_{n}-\lambda +m)^{2}+4\lambda m}}{2m} , \notag
       	\end{equation}
       	\begin{equation}
       	w_{b_{0}}=\frac{(1-p)w_{1}-(1-q)w_{0}}{q-p} , \notag
       	\end{equation}
       }
       	\indent The associated optimal $n$-year term life insurance purchasing strategy is not to purchase until the wealth reaches $ w_{b_{0}}$, and if $ w_{b_{0}} \leqslant w< w_{1}$, it's optimal to buy $n$-year term life insurance of $ f-w$; if $ w_{1} \leqslant w< w^{i_{0}}$, it's optimal to buy $n$-year term life insurance of 
       	$ f-\frac{H_{n}f}{r+a-c+H_{n}}$.\\
       	\indent The optimal investment strategy in risky market is not to invest until the wealth reaches $ w_{0}+\frac{\sigma l(1-p)}{\mu -r}$ and when the wealth is in a small range, i.e. $w_{1}+\frac{\sigma l(1-q)}{\mu -r} \leqslant w<w_{2}+\frac{\sigma l(1-p)}{\mu -r} $, the individual also doesn't invest. In other cases, the optimal investment amount is equal to
       	\begin{equation}
       	\pi(w)=\begin{cases}
       	\displaystyle \frac{(\mu -r)(w-w_{0})+\sigma l(p-1)}{\sigma ^{2}(1-p)} ,&w_{0}+\frac{\sigma l(1-p)}{\mu -r}\leqslant w < w_{b_{0}},\\
       	\displaystyle\frac{(\mu -r)(w-w_{1})+\sigma l(q-1)}{\sigma ^{2}(1-q)} ,&w_{b_{0}}\leqslant w< w_{1}+\frac{\sigma l(1-q)}{\mu -r} ,\\
       	\displaystyle\frac{(\mu -r)(w-w_{2})+\sigma l(p-1)}{\sigma ^{2}(1-p)} ,& w_{2} +\frac{\sigma l(1-p)}{\mu -r}  \leqslant w< w^{i_{0}},
       	\end{cases}	 \notag
       	\end{equation}
       \end{proposition}
       \begin{proof}
       	\rm Our general guideline is to use Lemma \ref{l1} to prove this proposition. In fact, it's a direct application of Lemma \ref{l1}.\\
       	\indent By Lemma \ref{l1} and Step 2, we can get three equations (\ref{e4}), (\ref{e5}), (\ref{e6}), by Lemma \ref{l2}, \ref{l3}, \ref{l4} and the continuity at critical points, we obtain $ \Phi_{n}$ in Proposition \ref{p1}.
     Notice that $\Phi_{n}(w)$ is non-decreasing continuous, and piecewise differentiable on $ [0,w^{i_{0}}) $, except at $ 0$, $ w_{0}$, $ w_{1}$ and $ w_{2}$. $ \Phi_{n}$ satisfies $ (\Phi_{n}) _{w} >0$, $ (\Phi_{n})_{ww}<0$ on $ (0,w^{i_{0}})\setminus \bigl((0,w_{0}) \cup (w_{1},w_{2}) \bigr)$, and obviously, $ F(0)=F(w_{0})=0,  F(w_{1})=F(w_{2})=1-e^{-\lambda n}, F(w^{i_{0}})=1$. After the above analysis, the proposition is the direct application of Lemma \ref{l1}. It is important to note that, since the investment amount in risky market is non-negative, to ensure that it's meaningful, we make adjustments to the range of $ \pi(w)$.
       \end{proof}
       \subsection{$n$-year term life insurance purchased by a continuously paid premium under the model \uppercase\expandafter{\romannumeral2}}
The specific components of model \uppercase\expandafter{\romannumeral2} are as follows.
       \begin{equation}
       \begin{cases}
       dS(t)=\mu S(t)dt+\sigma S(t)dB(t),\\
       dY(t)=adt+ldB(t),\\
       c(W)=c,
       \end{cases} \notag
       \end{equation}
       in which $ Y=\left\{Y(t)\right\}_{t\geqslant0}$ represents the price process of income, $ c$ represents the consumption rate. $ B=\left\{B(t)\right\}_{t\geqslant0}$ is a standard Brownian motion on a filtered probability space $ (\Omega ,\mathcal{F}, \mathbb{F}=\left\{\mathcal{F}_{t}\right\}_{t\geqslant 0},\mathbb{P})$, $\mu $, $ \sigma$, $ a$, $ l$, $ c$ are all constants, $ \mu >r$.
       Therefore, the wealth follows the dynamics
       {\small 
       \begin{equation}
       \begin{cases}
       dW(t)=\bigl [ rW(t)+(\mu -r)\pi _{t}-H_{n}D(t)\mathbbm{1}_{\left \{ t\leqslant n \right \}} +a-c\bigr ]dt+(\sigma \pi _{t}+l)dB_{t},&\ 0\leqslant t<\delta _{d},\\W(\delta _{d})=W(\delta _{d}-)+D(\delta _{d}-)\mathbbm{1}_{\left \{ \delta _{d}\leqslant n \right \}}.
       \end{cases} \notag
       \end{equation}
    }    
       \subsubsection{ The case for which $ a-c=0$}
       In this case, the problem is equivalent to the case where $ a-c=0$ in previous section, so we directly give the conclusion. The ``quasi-ideal value'' $w^{q_{1}}=\frac{H_{n}f}{r+H_{n}} $ and the ``ideal value'' $ w^{i_{1}}=\frac{\bigl(r+H_{n}+rH_{n}\bigr)f}{(r+H_{n})(r+1)}$.
       \begin{proposition}\label{p2}
       	The maximum probability of achieving the financial goal $ f$ before ruining is given by 
       	\begin{equation}
       	\Phi_{n}(w)=\begin{cases}
       	\displaystyle 0,&0 \leqslant w <w_{3} ,\\
       	\displaystyle (1-e^{-\lambda n})\frac{q(1-p)}{q-p}\left ( \frac{w-w_{3} }{w_{b_{1}}-w_{3}} \right )^{p} ,&w_{3}\leqslant w < w_{b_{1}},\\
       	\displaystyle (1-e^{-\lambda n})\left ( 1-\frac{p(q-1)}{q-p}\Biggl(\frac{w-w_{4}}{w_{b_{1}}-w_{4}}\Biggr)^{q} \right ),&w_{b_{1}} \leqslant w< w_{4},\\
       	\displaystyle 1-e^{-\lambda n},&w_{4} \leqslant w<w_{5}  ,\\
       	\displaystyle 1-e^{-\lambda n}+e^{-\lambda n} \left ( \frac{w- w_{2}}{w^{i_{1}}-w_{2}} \right )^{p},& w_{5}  \leqslant w< w^{i_{1}},
       	\end{cases}	 \notag
       	\end{equation}
       	in which 
       	\begin{equation}
       	w_{3}=\frac{(\mu -r)l}{\sigma r}, w_{4}=\frac{H_{n}f}{r+H_{n}}+\frac{(\mu -r)l}{\sigma (r+H_{n})}, 
       	w_{5}=\frac{H_{n}f}{r+H_{n}}+\frac{(\mu -r)l}{\sigma r}, \notag 
       	\end{equation}
       	\begin{equation}
       	p= \frac{k_{3}}{k_{3}-1}, \quad k_{3}=\frac{r-\lambda +m-\sqrt{(r-\lambda +m)^{2}+4\lambda m}}{2m} \notag
       	\end{equation}
       	\begin{equation}
       	q= \frac{k_{4}}{k_{4}-1}, \quad k_{4}=\frac{r+H_{n}-\lambda +m+\sqrt{(r+H_{n}-\lambda +m)^{2}+4\lambda m}}{2m} , \notag
       	\end{equation}
       	\begin{equation}
       	w_{b_{1}}=\frac{(1-p)w_{4}-(1-q)w_{3}}{q-p} , \notag
       	\end{equation}
       	\indent The associated optimal $n$-year term life insurance purchasing strategy is not to purchase until the wealth reaches $ w_{b_{1}}$, and if $ w_{b_{1}} \leqslant w< w_{4}$, it's optimal to buy $n$-year term life insurance of $ f-w$; if $ w_{4} \leqslant w< w^{i_{1}}$, it's optimal to buy $n$-year term life insurance of 
       	$ f-\frac{H_{n}f}{r+H_{n}}$.\\
       	\indent The optimal investment strategy in risky market is not to invest until the wealth reaches $ w_{3}+\frac{\sigma l(1-p)}{\mu -r}$ and when the wealth is in a small range, i.e. $w_{4}+\frac{\sigma l(1-q)}{\mu -r} \leqslant w<w_{5}+\frac{\sigma l(1-p)}{\mu -r} $, the individual also doesn't invest. In other cases, the optimal investment amount is equal to
       	\begin{equation}
       	\pi(w)=\begin{cases}
       	\displaystyle \frac{(\mu -r)(w-w_{3})+\sigma l(p-1)}{\sigma ^{2}(1-p)} ,&w_{3}+\frac{\sigma l(1-p)}{\mu -r}\leqslant w < w_{b_{1}},\\
       	\displaystyle\frac{(\mu -r)(w-w_{4})+\sigma l(q-1)}{\sigma ^{2}(1-q)} ,&w_{b_{1}} \leqslant w< w_{4}+\frac{\sigma l(1-q)}{\mu -r} ,\\
       	\displaystyle\frac{(\mu -r)(w-w_{5})+\sigma l(p-1)}{\sigma ^{2}(1-p)} ,& w_{5} +\frac{\sigma l(1-p)}{\mu -r}  \leqslant w< w^{i_{1}},
       	\end{cases}	 \notag
       	\end{equation}
       \end{proposition}
       \subsubsection{The case for which $ a-c>0$}
       In this case, we firstly give the new ``quasi-ideal value" $w^{q_{2}}=\frac{H_{n}f+c-a}{r+H_{n}} $ and the new ``ideal value" $w^{i_{2}}=\frac{(r+H_{n}+rH_{n})f-(c-a)H_{n}}{(r+H_{n})(r+1)} $, which solving methods are the same as that in the model \uppercase\expandafter{\romannumeral1}. Denote $ w_{b_{2}}$ which is called the buy level, then the control equation about $ \Phi_{n}$ as follows:
       {\small
       \begin{align}
      & \lambda \Bigl[ \Phi_{n} -(1-e^{-\lambda n})\mathbbm{1}_{\left\{ w\geqslant w_{b_{2}}\right\}} \Bigr]
       =\Bigl[ rw-H_{n}\Bigl(f-w\mathbbm{1}_{\left\{ w<\frac{H_{n}f+c-a}{r+H_{n}}+\frac{(\mu -r)l}{\sigma (r+H_{n})}\right\}}\notag\\-&\frac{H_{n}f}{r+a-c+H_{n}}\mathbbm{1}_{\left\{ w\geqslant \frac{H_{n}f+c-a}{r+H_{n}}+\frac{(\mu -r)l}{\sigma r} \right\}}\Bigr)\mathbbm{1}_{\left\{ w_{b_{2}}\leqslant w\leqslant w^{i_{2}}\right\}}+(c-a) \Bigr](\Phi_{n} )_{w}\notag\\+&\max_{\pi }\Bigl[ (\mu -r)\pi (\Phi_{n} )_{w}+\frac{1}{2}(\sigma \pi +l)^{2}(\Phi_{n} )_{ww} \Bigr] \notag
       \end{align}
    }
       Similarly, we simplify the above control equation into three parts.\\
       \noindent (\textbf{\romannumeral1}). If $ w<w_{b_{2}}$ (it can be seen from the solution of the latter equation that in fact $\frac{(\mu -r)l}{\sigma r}+\frac{c-a}{r} \leqslant w <w_{b_{2}} $), then the control equation can be replaced with the following equivalent expression
       \begin{equation}
       \lambda \Phi_{n}=\Bigl[rw-\frac{(\mu -r)l}{\sigma }+a-c\Bigr](\Phi_{n} ) _{w}-\frac{1}{2}\Bigl(\frac{\mu -r}{\sigma }\Bigr)^{2}\frac{(\Phi_{n} ) _{w}^{2}}{(\Phi_{n} ) _{ww}}. \label{e12}
       \end{equation}   
       \noindent (\textbf{\romannumeral2}). If $ w_{b_{2}}\leqslant w<\frac{H_{n}f+c-a}{r+H_{n}}+\frac{(\mu -r)l}{\sigma (r+H_{n})}$, then the control equation can be replaced with the following equivalent expression
       {\small 
       \begin{equation}\label{e13}
       \lambda\Bigl( \Phi_{n} -(1-e^{-\lambda n})\Bigr)=\Bigl[(r+H_{n})w-(H_{n}f+\frac{(\mu -r)l}{\sigma })+a-c\Bigr](\Phi_{n} )_{w}-\frac{1}{2}\Bigl(\frac{\mu -r}{\sigma }\Bigr)^{2}\frac{(\Phi_{n} )_{w}^{2}}{(\Phi_{n} )_{ww}}. 
       \end{equation}
    }
       \noindent (\textbf{\romannumeral3}). If $ \frac{H_{n}f+c-a}{r+H_{n}}+\frac{(\mu -r)l}{\sigma r} \leqslant w<w^{i_{2}}$, then the control equation can be replaced with the following equivalent expression
       {\small 
       \begin{equation}\label{e14}
       \lambda\Bigl( \Phi_{n} -(1-e^{-\lambda n})\Bigr)=\Bigl[rw-(\frac{(r+a-c)H_{n}f}{r+a-c+H_{n}}+\frac{(\mu -r)l}{\sigma })+a-c\Bigr](\Phi_{n} ) _{w}-\frac{1}{2}\Bigl(\frac{\mu -r}{\sigma }\Bigr)^{2}\frac{(\Phi_{n} ) _{w}^{2}}{(\Phi_{n} ) _{ww}}. 
       \end{equation}
    }
       These observations lead to the following verification Lemma \ref{l5}, which proof is same as the Lemma \ref{l1}.
       \begin{lemma}\label{l5}
       	Define a differential operator  $ \mathscr{L}_{2}^{\pi , D}$: 
       	{\small 
       	\begin{equation}
       	\mathscr{L}_{2}^{\pi,D}\tilde{F}=\bigl[(r+a-c)w+(\mu-r)\pi-HD\bigr]\tilde{F} _{w}+\frac{1}{2}(\sigma \pi +l)^{2}\tilde{F} _{ww}-\lambda \bigl[\tilde{F}-\lambda(1-e^{-\lambda n})\mathbbm{1}_{\left \{ w+D\geqslant f \right \}}\bigr]. \notag
       	\end{equation}
       }
       	Let $ F=F(w)$ be a function that is non-decreasing, continuous, and piecewise differentiable on $ [0,w^{i_{2}}) $, except that $ F $ might not be differentiable at $ 0$, $ w_{6}$, $ w_{7}$ and $ w_{8}$. Suppose $ F$ satisfies $ F _{w} >0$, $ F _{ww}<0$ on $ (0,w^{i_{2}})\setminus \bigl((0,w_{6}) \cup (w_{7},w_{8}) \bigr)$. If $ F$ satisfies the following boundary-value problem on $ [0,w^{i_{2}}]$:
       	\begin{equation}
       	\begin{cases}
       	\max_{\pi,D\geqslant0 } \mathscr{L}_{2}^{\pi,D}F(w)=0 ,\\
       	F(0)=F(w_{6})=0, \quad F(w_{7})=F(w_{8})=1-e^{-\lambda n}, \quad F(w^{i_{2}})=1.
       	\end{cases}	 \notag
       	\end{equation}
       	Then, on $ [0,w^{i_{2}}]$, 
       	\begin{equation}
       	\Phi_{n}=F. \notag
       	\end{equation}
       	The optimal investment amount in risk market is 
       	\begin{equation}
       	\pi ^{*}_{t}=-\frac{\mu -r}{\sigma }^{2}\frac{(\Phi_{n}) _{w}(w^{*}_{t})}{(\Phi_{n}) _{ww}(w^{*}_{t})}, \notag
       	\end{equation}
       	in which $ w^{*}_{t} $ is the optimal control wealth value at time $ t$.
       	The optimal $n$-year term life insurance purchase amount is
       	\begin{equation}
       	D ^{*}_{t}=\Bigl(f-w\mathbbm{1}_{\left\{ w<w_{7 }\right\}}-\frac{H_{n}f+c-a}{r+H_{n}}\mathbbm{1}_{\left\{ w\geqslant w_{8} \right\}}\Bigr)\mathbbm{1}_{\left\{ w_{b_{2}}\leqslant w\leqslant w^{i_{2}}\right\}} , \notag
       	\end{equation}
       	in which 
       	{\small 
       	\begin{equation}
       	w_{b_{2}}=\rm{inf}\left\{w\geqslant 0: \lambda(1-e^{-\lambda n})-H_{n}\Bigl( f-w\mathbbm{1}_{\left\{ w<w_{7 }\right\}}-\frac{H_{n}f+c-a}{r+H_{n}}\mathbbm{1}_{\left\{ w\geqslant w_{8} \right\}}\Bigr)F_{w}(w) \geqslant0 \right\} \wedge f, \notag
       	\end{equation}
       }
       	\begin{equation}
       	w_{6}=\frac{(\mu -r)l}{\sigma r}+\frac{c-a}{r}, w_{7}=\frac{H_{n}f+c-a}{r+H_{n}}+\frac{(\mu -r)l}{\sigma (r+H_{n})}, 
       	w_{8}=\frac{H_{n}f+c-a}{r+H_{n}}+\frac{(\mu -r)l}{\sigma r}. \notag 
       	\end{equation}
       \end{lemma}
       Since the solution of the maximum probability is exactly similar to the previous section, we directly give the following conclusions.
       \begin{proposition}\label{p3}
       	The maximum probability of achieving the financial goal $ f$ before ruining is given by 
       	\begin{equation}
       	\Phi_{n}(w)=\begin{cases}
       	\displaystyle 0,&0 \leqslant w <w_{6} ,\\
       	\displaystyle (1-e^{-\lambda n})\frac{q(1-p)}{q-p}\left ( \frac{w-w_{6} }{w_{b}-w_{6}} \right )^{p} ,&w_{6}\leqslant w < w_{b_{2}},\\
       	\displaystyle (1-e^{-\lambda n})\left ( 1-\frac{p(q-1)}{q-p}\Biggl(\frac{w-w_{7}}{w_{b}-w_{7}}\Biggr)^{q} \right ),&w_{b_{2}} \leqslant w< w_{7},\\
       	\displaystyle 1-e^{-\lambda n},&w_{7} \leqslant w<w_{8}  ,\\
       	\displaystyle 1-e^{-\lambda n}+e^{-\lambda n} \left ( \frac{w- w_{8}}{w^{i_{2}}-w_{8}} \right )^{p},& w_{8}  \leqslant w< w^{i_{2}},
       	\end{cases}	 \notag
       	\end{equation}
       	in which 
       	\begin{equation}
       	w_{6}=\frac{(\mu -r)l}{\sigma r}+\frac{c-a}{r}, w_{7}=\frac{H_{n}f+c-a}{r+H_{n}}+\frac{(\mu -r)l}{\sigma (r+H_{n})}, 
       	w_{8}=\frac{H_{n}f+c-a}{r+H_{n}}+\frac{(\mu -r)l}{\sigma r}, \notag 
       	\end{equation}
       	\begin{equation}
       	p= \frac{k_{3}}{k_{3}-1}, \quad k_{3}=\frac{r-\lambda +m-\sqrt{(r-\lambda +m)^{2}+4\lambda m}}{2m} \notag
       	\end{equation}
       	\begin{equation}
       	q= \frac{k_{4}}{k_{4}-1}, \quad k_{4}=\frac{r+H_{n}-\lambda +m+\sqrt{(r+H_{n}-\lambda +m)^{2}+4\lambda m}}{2m} , \notag
       	\end{equation}
       	\begin{equation}
       	w_{b_{2}}=\frac{(1-p)w_{7}-(1-q)w_{6}}{q-p} , \notag
       	\end{equation}
       	\indent The associated optimal $n$-year term life insurance purchasing strategy is not to purchase until the wealth reaches $ w_{b_{2}}$, and if $ w_{b_{2}}\leqslant w< w_{7}$, it's optimal to buy $n$-year term life insurance of $ f-w$; if $ w_{7} \leqslant w< w^{i_{2}}$, it's optimal to buy $n$-year term life insurance of 
       	$ f-\frac{H_{n}f+c-a}{r+H_{n}}$.\\
       	\indent The optimal investment strategy in risky market is not to invest until the wealth reaches $ w_{6}+\frac{\sigma l(1-p)}{\mu -r}$ and when the wealth is in a small range, i.e. $w_{7}+\frac{\sigma l(1-q)}{\mu -r} \leqslant w<w_{8}+\frac{\sigma l(1-p)}{\mu -r} $, the individual also doesn't invest. In other cases, the optimal investment amount is equal to
       	\begin{equation}
       	\pi(w)=\begin{cases}
       	\displaystyle \frac{(\mu -r)(w-w_{6})+\sigma l(p-1)}{\sigma ^{2}(1-p)} ,&w_{6}+\frac{\sigma l(1-p)}{\mu -r}\leqslant w < w_{b_{2}},\\
       	\displaystyle\frac{(\mu -r)(w-w_{7})+\sigma l(q-1)}{\sigma ^{2}(1-q)} ,&w_{b_{2}} \leqslant w< w_{7}+\frac{\sigma l(1-q)}{\mu -r} ,\\
       	\displaystyle\frac{(\mu -r)(w-w_{8})+\sigma l(p-1)}{\sigma ^{2}(1-p)} ,& w_{8} +\frac{\sigma l(1-p)}{\mu -r}  \leqslant w< w^{i_{2}},
       	\end{cases}	 \notag
       	\end{equation}
       \end{proposition}
       \subsubsection{The case for which $ a-c<0$}
       In this case, the control equation about $ \Phi_{n}$, the ``quasi-ideal value" $ w^{q_{2}}$ and the ``ideal value" $w^{i_{2}}$ are the same as in the previous section. We only consider the case when the consumption rate is large enough but less than a critical value, and the premium rate is small enough. \\
       \indent We first calculate the critical value of the consumption rate, denote this critical value is $ C_{0}$, then we obtain
       \begin{equation}
       \frac{\bigl(r+H_{n}+rH_{n}\bigr)f-(C_{0}-a)H_{n}}{(r+H_{n})(r+1)} =\frac{C_{0}-a}{r}. \notag
       \end{equation}
       Therefore, 
       \begin{equation}
       C_{0}=\frac{\bigl(r+H_{n}+rH_{n}\bigr)rf}{(r+H_{n})(r+1)+rH_{n}}+a. \notag
       \end{equation}
       \begin{remark}
       	{\rm(1)}. If $ c>C_{0}$, that means the consumption is too high, the problem is going to get more complicated and we won't discuss it in this paper.\\	
       	{\rm(2)}. If $ c \leqslant C_{0}$, due to the existence of noise terms $ l$ in the income process, we need to consider two cases $ l=0$ and $ l \neq 0$. 
       	\begin{itemize}
       	\item If $ l \neq 0$, $ w_{b_{2}}$ can't be equal to zero, then the associated conclusions are same as Proposition \ref{p3};
       	\item If $ l = 0$, $ w_{b_{2}}$ can be equal to zero, in this case, if the consumption rate is large enough and the premium rate is small enough, assume $ c \geqslant C_{1}$, then the individual purchase $ n$-year term life insurance at all wealth level.
       	The associated conclusions can be obtained by taking $ w_{b_{2}}=0$ and $ l=0$ in the Proposition \ref{p3}. For specific conclusions, see the following Proposition \ref{p4}.
       \end{itemize}
       \end{remark}
       \begin{proposition}\label{p4}
       	If $ l=0$, $ C_{1}+a \leqslant c \leqslant C_{0} $, $ H_{n} \leqslant \tilde{H}$, in which 
       	\begin{equation}
       	C_{1}=H_{n}f\Bigl(\frac{(r+H_{n})q}{\lambda}-1\Bigr), \notag
       	\end{equation}
       	$H_{n} \leqslant \tilde{H}$ is the solution of the following inequality
       	\begin{equation}
       	H_{n}f\Bigl(\frac{(r+H_{n})q}{\lambda}-1\Bigr) \leqslant\frac{\bigl(r+H_{n}+rH_{n}\bigr)rf}{(r+H_{n})(r+1)+rH_{n}}. \notag
       	\end{equation}
       	Then, the maximum probability of achieving the financial goal $ f$ before ruining is given by 
       	\begin{equation}
       	\Phi_{n}(w)=\begin{cases}
       	\displaystyle (1-e^{-\lambda n})\left ( 1-\Biggl(1-\frac{w}{w_{9}}\Biggr)^{q} \right ),&0 \leqslant w< w_{9},\\
       	\displaystyle 1-e^{-\lambda n}+e^{-\lambda n} \left ( \frac{w- w_{9}}{w^{i_{2}}-w_{9}} \right )^{p},& w_{9}  \leqslant w< w^{i_{2}},
       	\end{cases}	 \notag
       	\end{equation}
       	in which 
       	\begin{equation}
       	w_{9}=\frac{H_{n}f+c-a}{r+H_{n}},\notag 
       	\end{equation}
       	\begin{equation}
       	p= \frac{k_{3}}{k_{3}-1}, \quad k_{3}=\frac{r-\lambda +m-\sqrt{(r-\lambda +m)^{2}+4\lambda m}}{2m} \notag
       	\end{equation}
       	\begin{equation}
       	q= \frac{k_{4}}{k_{4}-1}, \quad k_{4}=\frac{r+H_{n}-\lambda +m+\sqrt{(r+H_{n}-\lambda +m)^{2}+4\lambda m}}{2m}. \notag
       	\end{equation}
       	\indent The associated optimal $n$-year term life insurance purchasing strategy is purchasing $n$-year term life insurance of $ f-w$ if $0 \leqslant w< w_{9}$; if $ w_{9} \leqslant w< w^{i_{2}}$, it's optimal to buy $n$-year term life insurance of 
       	$ f-\frac{H_{n}f+c-a}{r+H_{n}}$.\\
       	\indent The optimal investment amount is equal to
       	\begin{equation}
       	\pi(w)=\begin{cases}
       	\displaystyle\frac{(\mu -r)(w-w_{9})}{\sigma ^{2}(1-q)} ,&0 \leqslant w< w_{9} ,\\
       	\displaystyle\frac{(\mu -r)(w-w_{9})}{\sigma ^{2}(1-p)} ,& w_{9}  \leqslant w< w^{i_{2}},
       	\end{cases}	 \notag
       	\end{equation}
       \end{proposition}
       \begin{proof}
       	Our general guideline is to use lemma \ref{l5} to prove this proposition. Obviousely, $ (\Phi_{n} )_{w} >0$, $  (\Phi_{n} )_{ww}<0$ on $ (0, w^{i_{2}})$ and $ \Phi_{n}$ satisfies the boundary conditions, then we only need to prove  
       	\begin{equation}
       	\lambda(1-e^{-\lambda n})-H_{n}( f-w)  (\Phi_{n} )_{w}(w) \geqslant 0 \notag
       	\end{equation}
       	for all $ 0 \leqslant w \leqslant w^{i_{2}}$ .
       	As same as \cite{Bayraktar2016}, when $ c \geqslant C_{1} $, the above inequality holds.
       \end{proof}
       \begin{remark}
       	We give an explanation of $ C_{1}$ from another point of view.\\
       	If $w_{b_{2}}=0 $, then $ p=1$, $ q=\frac{\lambda}{r+H_{n}}$, and
       	\begin{equation}
       	1-\left ( 1-\frac{w}{\frac{C_{1}+H_{n}f}{r+H_{n}}} \right )^{q}=1-\left ( 1-\frac{w}{\frac{H_{n}f}{r+H_{n}}} \right )^{\frac{\lambda }{r+H_{n}}}, \notag
       	\end{equation}
       	\begin{equation}
       	\left ( 1-\frac{w}{\frac{H_{n}f}{r+H_{n}}} \right )^{\frac{1}{r+H_{n}}}=\left ( 1-\frac{w}{\frac{H_{n}fq}{\lambda}} \right )^{\frac{q}{\lambda}}	.\notag
       	\end{equation}
       	Combining the above two equations, we can get
       	\begin{equation}
       	C_{1}=H_{n}f\Bigl(\frac{(r+H_{n})q}{\lambda}-1\Bigr). \notag
       	\end{equation}
       \end{remark}

       \section{Summary and conclusions}\label{S5}
       We researched the problem about purchasing $m$-year deferred $n$-year term life insurance or $m$-year deferred $n$-year term pure endowment by a single premium or via a premium paid continuously in financial planning, to maximize the probability of achieving the given financial  goal $ f$. The problems were solved by giving the optimal purchasing strategies through establishing new deterministic control equations, ``\emph{quasi-ideal value}" and "\emph{ideal value}", and solving the associated variational inequalities.\\
       \indent  In Section \ref{S2}, it was shown that the optimal  purchasing strategy in the case where the policyholder purchases insurance by paying a single premium, is to buy it only when the asset has reached ``\emph{ideal value}". At that point, it is optimal to buy $m$-year deferred $n$-year term life insurance of $ f-D$, in which $ D$ represents the compensation available for other types of financial products that the policyholder owns before purchasing deferred term insurance. For the case of continuous premium payment, we also examined this situation by comparing the force of interest $r$ and the force of death $ \lambda$. We respectively gave the maximum probability of achieving the given financial goal and the related strategies for purchasing life insurance in the following two cases, if the policyholder has enough time to reach the ``\emph{ideal value}" (i.e. $\lambda \leqslant r$) and if he/she does not have enough time to reach the ``\emph{ideal value}"(i.e. $\lambda >r$). When $\lambda \leqslant r $, we got the similar conclusions as for the single premium; when $ \lambda >r$, the optimal purchasing strategy is to purchase insurance of $ f-w$ if $ w $ is less than the critical value $ w ^{0} \in (0,H^{*})$, otherwise, the policyholder does not buy insurance until the wealth reaches then ``\emph{ideal value}". In particular, if $m >0, n \rightarrow \infty$, our viewpoint also shed light on reaching a bequest goal by purchasing deferred whole life insurance. It is worth noting that if $m=0$, $ n \rightarrow \infty$, our problem is equivalent to achieving the just mentioned bequest goal by purchasing whole life insurance. In this case, the maximum probability and life insurance purchasing strategies we provided are consistent with those in \cite{Bayraktar2014}. \\
       \indent Besides the risk of death, there is another risk that deserves attention, namely the risk of longevity. To deal with this risk, we thought about the situation for which the policyholder purchased $m$-year deferred $n$-year term pure endowment in personal financial planning in Section \ref{S3}. We used some similar methods as in the previous section and got the corresponding optimal purchasing strategies, but the maximum probability of achieving the financial goal has changed. The associated results when $ m=0$ which mean the policyholder purchases $n$-year term pure endowment were given in the end of this section.\\
       \indent In Section \ref{S4}, we considered two models under the stochastic framework. In model \uppercase\expandafter{\romannumeral1}, the drift term about price process of income and the consumption rate correlated with the wealth value, the amount of investment in risky market followed the geometric Brownian motion. In this situation, we found the optimal $ n$-year term life insurance purchasing strategy is not to purchase until the wealth reaches the ``quasi-ideal value", and the associated optimal investment strategy in risky market is not to invest until the wealth reaches a relatively large value. We also found an interesting thing, when the wealth reaches a relatively large value but is in a small range, the individual also doesn't invest, see Proposition \ref{p1}. In model \uppercase\expandafter{\romannumeral2}, the amount of investment in risky market also followed the geometric Brownian motion, but the drift term about price process of income and the consumption rate were constants $ a$ and $ c$. Due to the change of the model, we discussed the above two constants in three cases. Specifically, the following results were obtained:
       \begin{itemize}    
 \item If $ a-c=0$, the problem was equivalent to the case where $ a-c=0$ in Section \ref{S2}, see Proposition \ref{p2};
 \item If $ a-c >0$, in this case, the idea of the problem and the solution methods were similar to those in the previous sections, but the ``quasi-ideal value", ``ideal value"  and the relevant critical values had changed, see Proposition \ref{p3};
  \item If $ a-c<0$, we considered the case when the consumption rate was large enough and the premium rate was small enough. If $ l \neq 0$, $ w_{b_{2}}$ can't be equal to zero, then the associated conclusions are same as Proposition \ref{p3}. If $ l = 0$, $ w_{b_{2}}$ can be equal to zero, in this case, if the consumption rate is large enough,then the individual purchase $ n$-year term life insurance at all wealth level, see Proposition \ref{p4}.
    \end{itemize}  
  
       \indent In our future work, on the one hand, we expect to complete the results in Section \ref{S4} when $ c>C_{0}$. On the other hand, we will extend our work to consider the individual purchasing irreversible insurance or $m$-year deferred $n$-year term life insurance with income and consumption to maximize the probability of achieving the financial goal, then find the optimal purchasing strategies in above cases. We expect this to be an important guide for personal financial management.
       \section*{Appendix}
       \noindent   \textbf{The specific method of determination about $D_{1}$, $D_{2} $, $ D_{3}$}.\\
       We first give the form of the solution of the following equation, the method refers to \cite{AV2018}, \cite{VI2020}:
       \begin{equation}
       y=xy'_{x}+a(y'_{x})^{n}, \notag
       \end{equation}
       \begin{equation}
       y=Ax^{\frac{n}{n-1}}, \quad aA^{n-1}n^{n}=-(n-1)^{n-1}, \quad n\neq 1. \notag
       \end{equation}
       Therefore, by Lemma \ref{l2} and Lemma \ref{l3}, we have
       \begin{equation}
       Dx_{1}(D_{1}p)^{x_{1}-1}=Ck_{1}(D_{2}q)^{k_{1}-1}=-1, \label{ea1}
       \end{equation}   
       in which $ C$, $ D$ are the constants to be determined.
       By the continuous of $ w_{b_{0}}$, 
       \begin{equation}
       D_{1}(w_{b_{0}}-w_{0})^{p}=1-e^{-\lambda n}+D_{2}(w_{b_{0}}-w_{1})^{q}, \notag
       \end{equation}
       thus, 
       \begin{equation}
       D_{2}= \frac{D_{1}(w_{b_{0}}-w_{0})^{p}-(1-e^{-\lambda n})}{(w_{b_{0}}-w_{1})^{q}} \label{ea2}
       \end{equation}
       Rewrite (\ref{ea1}) and by (\ref{ea2}), we obtain
       \begin{equation}
       D(x_{1}-1)D_{1}^{x_{1}-1}p^{x_{1}}=C(k_{1}-1)q^{k_{1}}\frac{\left[D_{1}(w_{b_{0}}-w_{0})^{p}-(1-e^{-\lambda n})\right]^{k_{1}-1}}{(w_{b_{0}}-w_{1})^{k_{1}}}, \notag
       \end{equation}
       i.e.
       \begin{equation}
       D_{1}^{x_{1}-1}=\frac{k_{1}-1}{x_{1}-1}\frac{q^{k_{1}}}{p^{x_{1}}}\left[D_{1}(w_{b_{0}}-w_{0})^{p}-(1-e^{-\lambda n})\right]^{k_{1}-1}\frac{C}{D}\frac{1}{(w_{b_{0}}-w_{1})^{k_{1}}}. \notag
       \end{equation}
       Therefore, we have the following form about $ D_{1}$ and $\displaystyle \frac{C}{D}$
       \begin{equation}
       D_{1}=R\frac{1}{(w_{b_{0}}-w_{0})^{p}}	,\quad  \frac{C}{D}=(w_{b_{0}}-w_{1})^{k_{1}}(w_{b_{0}}-w_{0})^{x_{1}}, \notag
       \end{equation}
       in which $ R$ is the solution of equation
       \begin{equation}
       R^{x_{1}-1}=\frac{k_{1}-1}{x_{1}-1}\bigl(R-(1-e^{-\lambda n})\bigr)^{k_{1}-1}\frac{q^{k_{1}}}{p^{x_{1}}}. \notag
       \end{equation}
       Then, we have
       \begin{equation}
       D_{1}=(1-e^{-\lambda n})\frac{q(1-p)}{q-p}\left ( \frac{1}{w_{b_{0}}-w_{0}} \right )^{p}, \notag
       \end{equation}
       \begin{equation}
       D_{2}=(1-e^{-\lambda n})\frac{p(1-q)}{q-p}(\frac{1}{w_{b_{0}}-w_{1}})^{q} , \notag
       \end{equation}
       $ D_{3}$ can be obtained in the same way as above
       \begin{equation}
       D_{3}=e^{-\lambda n} \left ( \frac{1}{w^{i_{0}}-w_{2}} \right )^{p}. \notag 
       \end{equation}$\hfill\blacksquare $

\bibliography{references}

\begin{thebibliography}{32}
\providecommand{\natexlab}[1]{#1}
\providecommand{\url}[1]{\texttt{#1}}
\expandafter\ifx\csname urlstyle\endcsname\relax
  \providecommand{\doi}[1]{doi: #1}\else
  \providecommand{\doi}{doi: \begingroup \urlstyle{rm}\Url}\fi

\bibitem[Andrei and Valentin(2018)]{AV2018}
D.~P. Andrei and F.~Z. Valentin.
\newblock \emph{Handbook of ordinary differential equations}.
\newblock CRC, 2018.

\bibitem[Arnold(2020)]{VI2020}
V.~I. Arnold.
\newblock \emph{Mathematical methods of classical mechanics (second edition)}.
\newblock Springer-Verlag, 2020.

\bibitem[Arup(2017)]{Arup2017}
S.~K. Arup.
\newblock Analysis of individual investors behavior of stock market.
\newblock \emph{IJTSRD}, 1\penalty0 (5):\penalty0 922--931, 2017.

\bibitem[Bajtelsmit and Wang(2018)]{VL2018}
V.~L. Bajtelsmit and T.~Y. Wang.
\newblock Household financial planning strategies for managing longevity risk.
\newblock \emph{FPR}, 1\penalty0 (2):\penalty0 e1007, 2018.

\bibitem[Bayraktar and Young(2006)]{Bayraktar2006}
E.~Bayraktar and V.~R. Young.
\newblock Minimizing the probability of lifetime ruin under borrowing
  constraints.
\newblock \emph{Insurance: Math. Econ.}, 41:\penalty0 196--221, 2006.

\bibitem[Bayraktar and Young(2013)]{B2013}
E.~Bayraktar and V.~R. Young.
\newblock Life insurance purchasing to maximize utility of household
  consumption.
\newblock \emph{North Amer. Actuarial J}, 17\penalty0 (2):\penalty0 114--135,
  2013.

\bibitem[Bayraktar et~al.(2014)Bayraktar, Promislow, and Young]{Bayraktar2014}
E.~Bayraktar, S.~D. Promislow, and V.~R. Young.
\newblock Purchasing life insurance to reach a bequest goal.
\newblock \emph{Insurance : Math. Econ.}, 58:\penalty0 204--216, 2014.

\bibitem[Bayraktar et~al.(2015)Bayraktar, Promislow, and Young]{Bayraktar2015}
E.~Bayraktar, S.~D. Promislow, and V.~R. Young.
\newblock Purchasing life insurance to reach a bequest goal:time-dependent
  case.
\newblock \emph{North Amer. Actuarial J.}, 19:\penalty0 224--236, 2015.

\bibitem[Bayraktar et~al.(2016)Bayraktar, Promislow, and Young]{Bayraktar2016}
E.~Bayraktar, S.~D. Promislow, and V.~R. Young.
\newblock Purchasing life insurance to reach a bequest goal while consuming.
\newblock \emph{SIAM J. Financial Math.}, 7:\penalty0 183--214, 2016.

\bibitem[Bender et~al.(2022)Bender, James, and et~al]{Bender2022}
S.~Bender, J.~C. James, and D.~D. et~al.
\newblock Millionaires speak: What drives their personal investment decisions?
\newblock \emph{J. Financ. Econ.}, 146\penalty0 (1):\penalty0 305--330, 2022.

\bibitem[Biradar et~al.(2021)Biradar, Adinoto, and et~al]{Biradar2021}
L.~S. Biradar, N.~Adinoto, and B.~F. et~al.
\newblock Personal financial planning: An approach towards insurance buying
  behavior.
\newblock \emph{JCMT}, 12\penalty0 (4):\penalty0 209--220, 2021.

\bibitem[Deimena(2014)]{KD2014}
K.~Deimena.
\newblock Individual investor investment alternatives assessment criteria
  modelling.
\newblock \emph{ATP}, 16:\penalty0 114--128, 2014.

\bibitem[Dhanasekaran and Kumar(2016)]{DR2016}
M.~Dhanasekaran and N.~R. Kumar.
\newblock Individual investors' behaviour towards investment on stocks.
\newblock \emph{AJRSSH}, 6\penalty0 (6):\penalty0 1656--1667, 2016.

\bibitem[Drive et~al.(2018)Drive, Freudenberg, Brimble, and Hunt]{DFB2018}
T.~Drive, B.~Freudenberg, M.~Brimble, and K.~Hunt.
\newblock Insurance literacy in australia: Not knowing the value of personal
  insurance.
\newblock \emph{Financ. Plan. Res. J}, 4:\penalty0 53--75, 2018.

\bibitem[Huang(2016)]{H2016}
L.~S. Huang.
\newblock Personal financial planning for college graduates.
\newblock \emph{Tech. Invest.}, 7:\penalty0 123--134, 2016.

\bibitem[Lee(2021)]{L2021}
H.~S. Lee.
\newblock Life insurance and subsistence consumption with an exponential
  utility.
\newblock \emph{J. Mathematics.}, 9:\penalty0 358, 2021.

\bibitem[Liang and Young(2019)]{LY2019}
X.~Q. Liang and V.~R. Young.
\newblock Reaching a bequest goal with life insurance : ambiguity about the
  risky asset's drift and mortality's hazard rate.
\newblock \emph{ASTIN Bulletin: The Journal of the IAA}, 50\penalty0
  (1):\penalty0 187--221, 2019.

\bibitem[Liang and Zhao(2016)]{LZ2016}
Z.~X. Liang and X.~Y. Zhao.
\newblock Optimal investment, consumption and life insurance under stochastic
  framework(in chinese).
\newblock \emph{Sci. Sin. Math}, 46:\penalty0 1863--1882, 2016.

\bibitem[Merton(1969)]{RM1969}
R.~C. Merton.
\newblock Lifetime portfolio selection under uncertainty: The continuous-time
  case.
\newblock \emph{Rev. Econ. stat.}, 51:\penalty0 247--257, 1969.

\bibitem[Merton(1971)]{RM1971}
R.~C. Merton.
\newblock Optimum consumption and portfolio rules in a continuous-time model.
\newblock \emph{J. Econ. Theory}, 3:\penalty0 373--413, 1971.

\bibitem[Michael et~al.(2021)Michael, Jeffrey, and et~al]{MJ2021}
J.~N. Michael, T.~C. Jeffrey, and A.~J. et~al.
\newblock \emph{Personal Financial Planning for Executives and Entrepreneurs
  (first edition)}.
\newblock Springer, 2021.

\bibitem[Milevsky et~al.(2006)Milevsky, Moore, and Young]{MY2006}
M.~A. Milevsky, K.~S. Moore, and V.~R. Young.
\newblock Asset allocation and annuity-purchase strategies to minimize the
  probability of financial ruin.
\newblock \emph{Math. Finance.}, 16:\penalty0 647--671, 2006.

\bibitem[Pietrzyk and Rokita(2015)]{PR2015}
R.~Pietrzyk and P.~Rokita.
\newblock Stochastic goals in financial planning for a two-person household.
\newblock \emph{Stat. Transit}, 16:\penalty0 111--136, 2015.

\bibitem[Promislow(2011)]{DP2011}
S.~D. Promislow.
\newblock \emph{Fundamentals of Actuarial Mathematics (third edition)}.
\newblock John Wiley Sons, 2011.

\bibitem[Richard(1975)]{FR1975}
S.~F. Richard.
\newblock Optimal consumption, portfolio and life insurance rules for an
  uncertain lived individual in a continuous time model.
\newblock \emph{J. Fina. Econ}, 2:\penalty0 187--203, 1975.

\bibitem[Scriven(2008)]{DS2008}
D.~Scriven.
\newblock \emph{Guide to Life Protection and Planning (second edition)}.
\newblock CCH Australia Limited, 2008.

\bibitem[Subbakrishna and Murali(2018)]{SM2018}
K.~R. Subbakrishna and S.~Murali.
\newblock \emph{Personal Financial Planning(Wealth Management) (first
  edition)}.
\newblock Himalaya Publishing House Pvt. Ltd, 2018.

\bibitem[Topa et~al.(2018)Topa, Lunceford, and Boyatzis]{GG2018}
G.~Topa, G.~Lunceford, and R.~E. Boyatzis.
\newblock Financial planning for retirement: A psychosocial perspective.
\newblock \emph{Front. Psychol}, 8:\penalty0 2338, 2018.

\bibitem[Wang and Young(2012)]{WY2012}
T.~Wang and V.~R. Young.
\newblock Optimal commutable annuities to minimize the probability of lifetime
  ruin.
\newblock \emph{Insurance: Math. Econ}, 50:\penalty0 200--216, 2012.

\bibitem[Weedige et~al.(2019)Weedige, Ouyang, and et~al]{WOG2019}
S.~S. Weedige, H.~B. Ouyang, and Y.~G. et~al.
\newblock Decision making in personal insurance: Impact of insurance literacy.
\newblock \emph{Sustainability}, 11\penalty0 (23):\penalty0 6795, 2019.

\bibitem[Xiong and Shen(2020)]{XS2020}
F.~S. Xiong and Z.~Z. Shen.
\newblock \emph{Life insurance actuarial science (second edition)}.
\newblock Wuhan University Press, 2020.

\bibitem[Young(2004)]{Y2004}
V.~R. Young.
\newblock Optimal investment strategy to minimize the probability of lifetime
  ruin.
\newblock \emph{North Amer. Actuarial J}, 8\penalty0 (4):\penalty0 106--126,
  2004.

\end{thebibliography}
\end{document}